\documentclass[journal,twoside,web]{ieeecolor}
\usepackage{generic}
\usepackage{cite}
\usepackage{amsmath,amssymb,amsfonts}
\usepackage{algorithmic}
\usepackage{graphicx}
\usepackage{algorithm,algorithmic}
\usepackage{hyperref}
\hypersetup{hidelinks=true}
\usepackage{textcomp}
\def\BibTeX{{\rm B\kern-.05em{\sc i\kern-.025em b}\kern-.08em
    T\kern-.1667em\lower.7ex\hbox{E}\kern-.125emX}}
\usepackage{graphicx}
\usepackage{amsmath, amsfonts, mathtools, url, hyperref, subcaption, amssymb,comment}
\usepackage{cite}
\usepackage{cleveref}
\usepackage{graphicx}
\usepackage{tikz}
\usepackage{pgfplots}
\pgfplotsset{compat=1.18}
\usetikzlibrary{positioning, shapes.geometric}

\crefname{figure}{Fig.}{Figs.}
\crefname{example}{Example}{Examples}
\crefname{definition}{Definition}{Definitions}
\crefname{proposition}{Proposition}{Propositions}
\crefname{corollary}{Corollary}{Corollaries}
\crefname{theorem}{Theorem}{Theorems}
\crefname{lemma}{Lemma}{Lemmas}
\crefname{assumption}{Assumption}{Assumptions}
\crefname{section}{Section}{Sections}
\crefname{appendix}{Appendix}{Appendices}
\crefname{table}{Table}{Tables}
\crefname{footnote}{Footnote}{Footnotes}

\newtheorem{theorem}{Theorem}
\newtheorem{definition}{Definition}
\newtheorem{remark}{Remark}
\newtheorem{example}{Example}
\newtheorem{proposition}{Proposition}
\newtheorem{lemma}{Lemma}
\newtheorem{corollary}{Corollary}

\begin{document}
\title{Congestion Games with Heterogeneous Valuations: An Optimal Transport Approach}
\author{Pan-Yang Su, \IEEEmembership{Student Member, IEEE}, Negar Mehr, \IEEEmembership{Member, IEEE}, and Shankar Sastry, \IEEEmembership{Life Fellow, IEEE}
\thanks{Pan-Yang Su, Negar Mehr, and Shankar Sastry are with the University of California, Berkeley, CA 94720, USA (email: pan\_yang\_su@berkeley.edu).}
}
\maketitle

\begin{abstract}
In emerging urban mobility and logistics applications, such as advanced air mobility, electric vehicle charging, and shared service systems, agents with heterogeneous valuations choose among multiple outcomes while sharing congested resources. Existing congestion game and resource-allocation models, however, typically restrict either the structure of heterogeneous outcome preferences or the form of congestion externalities. We introduce a nonatomic congestion game framework in which agents have multidimensional heterogeneous valuations over a finite set of outcomes, with the population represented by a measure over the valuation space. We show that, although there is a continuum of heterogeneous agents,
both Nash equilibria and social optima admit
finite-dimensional characterizations. Nash equilibria are characterized by aggregate flows and outcome-level threshold vectors, while social optima admit a semi-discrete optimal transport formulation characterized by aggregate flows and strategy-level Kantorovich dual potentials. For
separable utilities and latency functions $\ell_e$ such that
$x\mapsto x\ell_e(x)$ is convex for every resource $e$, the social welfare problem becomes a convex optimization problem, and we characterize social optima through marginal congestion costs and capacity-dual prices. Our results provide a geometric characterization of congestion games with multidimensional heterogeneous valuations and connect nonatomic congestion games with semi-discrete optimal transport.
\end{abstract}


\section{Introduction}
\label{sec: introduction}

\subsection{A New Congestion Game Model}

Many applications involve a continuum of agents who choose among multiple outcomes while sharing congested resources. Examples include travelers choosing among alternative destinations, electric vehicle owners selecting charging stations, and users selecting service facilities such as delivery hubs or cloud servers. In these settings, agents may value different outcomes differently, while congestion affects the cost of each strategy and the available resources impose capacity constraints.

A common feature of these problems is the interaction between \emph{heterogeneous outcome preferences} and \emph{aggregate congestion}. Classical nonatomic congestion games capture congestion effects and strategic route choice, but typically assume that agents have homogeneous preferences over outcomes or that demand heterogeneity can be summarized by a one-dimensional valuation. Conversely, matching and resource-allocation models allow for heterogeneous preferences over a finite set of alternatives, but typically take the costs of assigning agents to alternatives as exogenous and therefore do not capture congestion generated by aggregate resource utilization.

We develop a nonatomic congestion game that combines these two features. Agents are characterized by multidimensional valuation vectors over a finite set of outcomes, and their utilities depend on both their outcome-specific valuations and congestion on the resources used by their chosen strategies. The model also accommodates hard resource capacities and more general constraints. This framework therefore provides a common formulation for congestion and allocation problems in which both the composition of the population assigned to each outcome and the congestion generated by that assignment are endogenous.

\subsection{Theoretical Contributions}

Our framework connects two classical lines of research: nonatomic congestion games \cite{10.5555/1076293,ROUGHGARDEN2004389} and semi-discrete matching and optimal transport \cite{wolansky2020semidiscreteoptimaltransport,doi:10.1086/687476}. When agents are homogeneous, our model reduces to a standard nonatomic congestion game. When congestion is absent, it becomes an assignment problem in which a continuum of heterogeneous agents is allocated to a finite set of outcomes.

The combination of multidimensional heterogeneous valuations and endogenous congestion creates a different structure from the standard one-dimensional variable-demand setting \cite{10.5555/1109557.1109630}. With a single outcome, heterogeneous demand can often be characterized by a scalar cutoff separating agents who participate from those who do not. With multiple outcomes, however, an agent's type is a vector of valuations, and there is generally no natural total ordering of types. The allocation is therefore described by a partition of the valuation space rather than by a single cutoff.

Our first result (\cref{prop: ne}) concerns Nash equilibria. We show directly from the agents' best-response conditions that, for a given aggregate flow, all active strategies associated with the same outcome must have the same latency. This yields a finite-dimensional characterization of equilibria in terms of aggregate flows and one threshold vector for each group. The threshold vector partitions the valuation space into regions corresponding to the outcomes chosen by different types. Thus, although an equilibrium is an infinite-dimensional action profile over a continuum of agents, it can be described by finitely many variables.

Our second result (\cref{prop: kantorovich 0}) concerns social optima. For a fixed aggregate flow, assigning heterogeneous agents to strategies is a semi-discrete optimal transport problem \cite{wolansky2020semidiscreteoptimaltransport}. Kantorovich duality (see, e.g., \cite[Theorem 5.10]{villani2008optimal}) then provides a finite-dimensional vector of dual potentials that characterizes the optimal assignment and induces a partition of the valuation space across strategies. The remaining optimization is over aggregate flows. Consequently, the equilibrium and social-optimum characterizations rely on different structures. Nash equilibria are characterized through agents' best-response conditions and outcome-level thresholds, whereas social optima exploit the primal-dual structure of a semi-discrete optimal transport problem.

Our third result (\cref{prop: opt}) specializes the model to separable utilities. In this setting, the social optimum admits a characterization through marginal congestion costs and dual variables associated with capacity constraints. This recovers the familiar marginal-cost and capacity-dual pricing structure of side-constrained traffic assignment models \cite{LARSSON1999233}, while incorporating heterogeneous destination valuations through the threshold vector.

\subsection{An Emerging Application}

One motivating application is \textbf{Advanced Air Mobility (AAM)}, an emerging air transportation paradigm envisioned by the Federal Aviation Administration (FAA) \cite{faaUtmConops,faaUamConops} to provide operations such as passenger transport \cite{Joby2025}, package delivery \cite{WingWalmart2025}, and other short-distance services. As the number of vehicles and available destinations grows, AAM systems must allocate agents across routes and destinations while accounting for congestion and limited infrastructure capacity. Existing air traffic management approaches primarily focus on aggregate traffic and delay \cite{doi:10.1287/opre.46.3.406,doi:10.1287/trsc.34.3.239.12300,doi:10.1287/opre.1100.0899,10.1007/978-3-642-86726-2_17}, while recent market-based approaches introduce mechanisms for eliciting heterogeneous valuations but often represent congestion through capacity constraints rather than congestion-induced delays \cite{balakrishnan2017distributed,10886225}. Our framework provides a model in which both effects are endogenous: agents may value different destinations differently, while their utilities depend on the congestion generated by the aggregate allocation of agents to the underlying network resources. 

In summary, our main contributions are as follows.
\begin{enumerate}
\item We introduce a nonatomic congestion game with multidimensional, heterogeneous outcome valuations and aggregate resource congestion.

\item Under our model, we establish finite-dimensional characterizations of Nash equilibria and social optima. Nash equilibria are characterized by aggregate flows and outcome-level threshold vectors. Social optima admit a semi-discrete optimal transport formulation characterized by aggregate flows and strategy-dependent Kantorovich dual potentials.

\item For separable utilities and latency functions $\ell_e$
such that $x\mapsto x\ell_e(x)$ is convex, we show that the social welfare problem is a convex optimization problem and characterize social optima using marginal congestion costs and capacity dual variables. We further show that social optima can be induced as Nash equilibria through element-specific marginal-cost and capacity prices.

\end{enumerate}

The rest of the paper is organized as follows. In \cref{sec: related work}, we review the related literature. We then formalize the model in \cref{sec: model} and discuss applications in \cref{sec: app}. \cref{sec: Characterizations} presents the main theoretical results on Nash equilibria and social optima. \cref{sec: separable} specializes these results to separable utilities. \cref{sec: Extensions} develops two extensions, and \cref{sec: Conclusion} concludes the paper. 

\emph{Notation:}
We denote the set of real numbers by \(\mathbb{R}\) and nonnegative real numbers by \(\mathbb{R}_+\). For a positive integer $N$, we define $[N] := \{1, 2, \ldots, N\}$. When indexing a vector $v = (v_1, \ldots, v_n) \in \mathbb{R}^n$, we follow the standard game-theoretic notation: $(a, v_{-i}) = (v_1, \ldots, v_{i-1}, a, v_{i+1}, \ldots, v_n)$ for any $i \in [n]$. When comparing two vectors $u, v \in \mathbb{R}^n$, we use $u \geq v$ to denote the pointwise comparison; that is, $u_i \geq v_i$ for any $i \in [n]$. The same convention applies to $>$, $<$, and $\leq$. For a mapping $f: X \rightarrow Y$ and a subset $V \subseteq Y$, we denote the inverse image of $V$ under $f$ as $f^{-1}(V) := \{x\in X: f(x) \in V\}$. When $V$ is a singleton, say $V = \{y\}$, we use the shorthand notation $f^{-1}(y)$ to denote $f^{-1}(\{y\})$. For two measurable spaces $(X,\mathcal{A})$ and $(Y,\mathcal{B})$,
a measurable mapping $f:X\to Y$, and a measure $\mu$ on
$(X,\mathcal{A})$, we denote the pushforward of $\mu$ under $f$ by $f_{\#}\mu$, defined by $(f_{\#}\mu)(B):=\mu\bigl(f^{-1}(B)\bigr)$ for any $B\in\mathcal{B}$. Given a measure space, we
abbreviate \emph{almost everywhere} by \emph{a.e.}

\section{Related Work}
\label{sec: related work}
Our work lies at the intersection of two lines of research: nonatomic congestion games \cite{10.5555/1076293,ROUGHGARDEN2004389} and matching markets with a continuum of agents \cite{gretsky1992nonatomic,GRETSKY199960,wolansky2020semidiscreteoptimaltransport,doi:10.1086/687476}.

\subsection{Nonatomic Congestion Games}
Our model extends nonatomic congestion games by incorporating
heterogeneous multidimensional valuations over different outcomes. Milchtaich
\cite{MILCHTAICH1996111,MILCHTAICH2009750,MILCHTAICH200456}
studies congestion games with player-specific preferences in which each strategy consists of a single resource. In contrast, our model allows each strategy to use multiple resources, so different strategies and outcomes may interact through shared congested resources. Moreover,
\cite{MILCHTAICH1996111,MILCHTAICH2009750} consider finitely many players, while \cite{MILCHTAICH200456} considers a continuum of players but assumes utilities that are separable in player-specific preferences and congestion. We allow general utility functions, and our continuum model with multidimensional heterogeneous valuations leads to finite-dimensional partition characterizations that are absent from
\cite{MILCHTAICH1996111,MILCHTAICH2009750,MILCHTAICH200456}.

When specialized to routing, our model generalizes routing games and traffic assignment models with elastic (variable) demand
\cite{10.5555/1109557.1109630,
10.1145/2000378.2000391,
YILDIRIM2005659,
doi:10.1287/trsc.14.2.174}. These models are recovered by including an additional destination
that represents the outside option of not traveling. Our model also generalizes routing models with capacity and side constraints \cite{LARSSON1999233,doi:10.1287/moor.1040.0098,LARSSON1995433,NIE2004285}. Also, unlike combined destination-choice and traffic assignment models \cite{EVANS197637,LAM1992275,doi:10.1287/trsc.22.1.14}, we do not assume a parametric random-utility specification. Instead, agents have arbitrary heterogeneous valuations over destinations. Recently, Tan et al. \cite{10666918} study joint destination and route choice in a nonatomic congestion game, where destination and route costs depend on aggregate flows. In contrast, our model also incorporates heterogeneous multidimensional valuations over destinations in addition to congestion costs.

\subsection{Matching Markets}
On the matching side, our work is related to assignment games \cite{shapley1971assignment} and their extensions to markets with a continuum of agents \cite{gretsky1992nonatomic,GRETSKY199960}, as well as semi-discrete matching and optimal transport models \cite{wolansky2020semidiscreteoptimaltransport}. Our work is particularly related to the framework of Azevedo and Leshno \cite{doi:10.1086/687476}, who derive cutoff-based allocations when matching a finite set of colleges and a continuum of students. Like their model, our equilibrium and welfare characterizations are described by finite-dimensional threshold vectors that partition a continuum of heterogeneous agents. However, Azevedo and Leshno study stable matchings in a two-sided market with ordinal preferences and exogenous matching values, whereas we study a one-sided congestion game with cardinal utilities and endogenous congestion externalities. 

Our work is also related to matching with externalities
\cite{KeisukeBando2016,10.1093/restud/rdac032,SASAKI199693},
which studies how interdependent preferences affect stability
and assignment mechanisms. Closely related work includes
Phan et al. \cite{10.1257/aer.20220626}, who study crowding
in school choice, and Estes and Mani
\cite{estes2023nonmonetary}, who study congestion in finite
matching markets with heterogeneous valuations.
In contrast, we develop a nonatomic congestion-game framework
that accommodates general utility functions and network
structures, and characterize Nash equilibria and social
optima rather than stable matchings.

\subsection{Optimal Transport for Nonatomic Games}
Blanchet and Carlier \cite{doi:10.1287/moor.2015.0719,blanchet2014remarks}
study Cournot--Nash equilibria in nonatomic games using optimal
transport and related methods. Their framework allows general type and strategy spaces, while their optimal-transport analysis primarily exploits a separable structure between type-dependent costs and aggregate interaction effects. Our model instead considers finite outcome and strategy spaces but allows utilities to depend nonseparably on heterogeneous valuations and endogenous congestion. The role of optimal transport is also different: their work uses optimal transport in the analysis of Cournot--Nash equilibria, whereas we use semi-discrete optimal transport to characterize socially optimal assignments conditional on aggregate flows. 

Carlier et al. \cite{doi:10.1137/060672832}
study optimal transportation with traffic congestion and characterize Wardrop equilibria in a continuous-space model. In their framework, optimal transport describes the physical movement of mass between prescribed source and target distributions through a congested space. In contrast, our optimal-transport formulation describes the assignment
of heterogeneous agents to outcomes and strategies, whose aggregate flows are endogenously determined.

\section{Model}
\label{sec: model}

In this section, we introduce a nonatomic congestion game with heterogeneous valuations. Our model  extends the classical nonatomic congestion games by introducing heterogeneous multidimensional valuations over a finite set of outcomes. The reader may consult \cref{ex: two 1} for a concrete example.

\subsection{Nonatomic Congestion Games with Heterogeneous Valuations}
\label{subsec: val and act}

\subsubsection{Elements}
Let $E$ be a finite ground set of \emph{elements}, where each element represents a resource, such as an edge in a traffic network or a vertiport in AAM. Each element $e \in E$ has a positive \emph{capacity} $c_e$ and a \emph{cost} or \emph{latency function} $\ell_e:\mathbb{R}_+ \rightarrow \mathbb{R}_+$. We assume that each latency function is nonnegative, nondecreasing, and continuous. In addition to element-wise capacity constraints, our model can accommodate general constraints, such as regional capacity constraints involving multiple elements; see \cref{sec: convex}.

\subsubsection{Agents, Outcomes, and Strategies}

There are $k$ \emph{groups} of agents. Each group $i \in [k]$ is associated with a positive \emph{rate} $r_i$, representing the total mass of agents in the group, a finite set of \emph{outcomes} $D_i$, and a finite \emph{strategy set} $\mathcal{P}_i \subseteq 2^E$. We assume that strategy sets are disjoint, i.e., $P_i \cap P_{i'} = \emptyset$ for $i \neq i'$.

Each strategy yields a particular outcome. Formally, the outcome set $D_i$ induces a partition $\{\mathcal{P}_{i,j}\}_{j \in D_i}$ of $\mathcal{P}_i$, where
\[
\mathcal{P}_i = \bigcup_{j\in D_i}\mathcal{P}_{i,j},
\qquad
\mathcal{P}_{i,j}\cap \mathcal{P}_{i,j'}=\emptyset
\quad \text{for } j\neq j',
\]
and $\mathcal{P}_{i,j}$ denotes the set of strategies that yield outcome $j$. Thus, a strategy specifies the collection of resources that an agent utilizes to obtain a particular outcome.

In routing games, a strategy corresponds to a path, outcomes correspond to physical destinations, and $\mathcal{P}_{i,j}$ is the set of paths available to group $i$ that lead to destination $j$; see \cref{sec: examples and applications}. We denote the joint strategy set by $\mathcal{P}:=\bigcup_{i\in[k]}\mathcal{P}_i$.

\subsubsection{Flows}
A \emph{flow} is a function $f:\mathcal{P}\rightarrow\mathbb{R}_+$, where $f_P$ denotes the mass of agents using strategy $P$. Since $\mathcal{P}$ is finite, a flow can equivalently be viewed as a $|\mathcal{P}|$-dimensional vector or as a finite measure on $\mathcal{P}$. These equivalent representations will be useful later.

For a given flow $f$, the load on element $e\in E$ is $f_e:=\sum_{P\ni e}f_P$. A flow $f$ is \emph{admissible} if $\sum_{P\in\mathcal{P}_i} f_P=r_i$ for every group $i \in [k]$. An admissible flow is \emph{capacity-feasible} if, in addition, $f_e\leq c_e$ for every element $e \in E$. Given a flow $f$, the latency of strategy $P$ is
\[
\ell_P(f):=\sum_{e\in P}\ell_e(f_e).
\]
Thus, $\ell_e(f_e)$ depends only on the load on element $e$, whereas $\ell_P(f)$ generally depends on the loads of all elements used by $P$.

\subsubsection{Always-Feasible Strategies}
We write $c_e=\infty$ whenever the capacity of element $e$ is nonbinding for every possible flow; in particular, this holds whenever
\[
c_e\geq
\sum_{i\in I_e}r_i,
\qquad
I_e:=\{i\in[k]\mid \exists P\in\mathcal{P}_i \text{ such that }e\in P\}.
\]
To ensure that every agent has access to a feasible strategy, we assume that each group $i$ has at least one strategy consisting entirely of infinite-capacity elements. Such a strategy may be interpreted as an \emph{outside option}, such as not participating in the system.

\subsubsection{Valuations}
Agents within a group may have different valuations for the available outcomes. For each group $i$, let $V_i:=[0,\bar V]^{|D_i|}$ for some $\bar V>0$. A vector $v\in V_i$ specifies an agent's valuations over the outcomes in $D_i$, with $v_j$ denoting the valuation for outcome $j$. We represent the population of group $i$ by a finite measure $\mu_i$ on the valuation space $V_i$.
\begin{definition}
The \emph{valuation space} of group $i$ is the finite measure space $(V_i,\mathcal{B}_{V_i},\mu_i)$, where $V_i=[0,\bar V]^{|D_i|}$, $\mathcal{B}_{V_i}$ is the Borel $\sigma$-algebra on $V_i$, and $\mu_i(V_i)=r_i$.
\end{definition}

Thus, $\mu_i$ describes how the mass of agents in group $i$ is distributed over valuation vectors. There is no randomness in the model: the measure $\mu_i$ represents the distribution of types in a continuum population rather than uncertainty about an individual agent's valuation.

We assume throughout the main analysis that $\mu_i$ is absolutely continuous with respect to the Lebesgue measure $\lambda$. We relax this assumption in \cref{subsec: abs conti}.

We identify each agent by a group-valuation pair $(i,v)$, where $i\in[k]$ and $v\in V_i$. The joint population is therefore represented by the disjoint union of the group-specific valuation spaces.

\begin{definition}
The \emph{joint valuation space} is the measure space $(V,\mathcal{B}_V,\mu)$, where
\[
V:=\bigsqcup_{i\in[k]}V_i,
\qquad
\mathcal{B}_V
:=
\left\{
\bigsqcup_{i\in[k]}A_i
\;\middle|\;
A_i\in\mathcal{B}_{V_i},\ \forall i\in[k]
\right\},
\]
and
\[
\mu\left(\bigsqcup_{i\in[k]}A_i\right)
:=
\sum_{i\in[k]}\mu_i(A_i).
\]
\end{definition}

\subsubsection{Action Profiles}

For an agent $(i,v)$, the action space is $\mathcal{P}_i$. Following the standard formulation of nonatomic games \cite{schmeidler1973equilibrium}, we represent an action profile as a measurable function assigning a strategy to each agent.

\begin{definition}
\label{def: action profile}
An \emph{action profile} is a measurable function $a:V\rightarrow\mathcal{P}$ such that
\[
a((i,v))\in\mathcal{P}_i,
\qquad
\forall i\in[k],\ v\in V_i.
\]
Equivalently, we write
\[
a:=(a_{i,v})_{i\in[k],\,v\in V_i},
\qquad 
a_{i,v}:=a((i,v)).
\]
\end{definition}

The measurability of $a$ ensures that the set $\{v\in V_i\mid a_{i,v}=P\}$ is $\mu_i$-measurable for every $i\in[k]$ and $P\in\mathcal{P}_i$, so the induced flows and aggregate utilities are well-defined.

Following \cite[Theorem 2]{schmeidler1973equilibrium}, it suffices to consider pure strategies. As is standard in nonatomic games, we also identify measurable action profiles that agree almost everywhere.

An action profile induces a flow through the pushforward of the
population measure, i.e, $f=a_{\#}\mu$, or equivalently,
\[
f_P
=
\mu\bigl(a^{-1}(P)\bigr),
\qquad
\forall P\in\mathcal{P}.
\]
For each group $i$, we denote by $f^i$ the restriction of $f$ to $\mathcal{P}_i$, viewed as a measure on $\mathcal{P}_i$.

The following lemma establishes the relationship between action profiles and admissible flows.
\begin{lemma}
\label{lemma: admissible flow}
Every action profile induces an admissible flow. Conversely, every admissible flow is induced by some action profile.
\end{lemma}

\begin{proof}
Every flow induced by an action profile is admissible by construction, so it remains to prove the converse. Fix a group $i$. Since $\mu_i$ is nonatomic, there exists a measurable partition
$\{A_P\}_{P\in\mathcal{P}_i}$ of $V_i$ such that $\mu_i(A_P)=f_P$ for every $P\in\mathcal{P}_i$ (cf. \cite[Theorem 10.52]{aliprantis2006infinite}). Defining $a_{i,v}=P$ for $v\in A_P$ therefore induces the desired flow. Repeating this construction for every group establishes the result.
\end{proof}

\subsubsection{Utility Functions}

Each group $i$ is associated with a continuous utility function
$U_i:\mathbb{R}_+^2\rightarrow\mathbb{R}$ that is strictly increasing in its first argument and strictly decreasing in its second argument. The first argument represents the agent's valuation for the realized outcome, while the second represents the congestion cost of the chosen strategy. The separable utility $U_i(x,y)=x-y$ is an important special case; see \cref{sec: examples and applications}.

Given an action profile $a$ with induced flow $f$, the utility of agent $(i,v)$ is
\begin{equation}
\label{eq: utility}
\begin{aligned}
&u_{i,v}(a_{i,v},a_{-(i,v)})\\
&=
\begin{cases}
U_i(v_j,\ell_P(f)),
&
a_{i,v} = P\in\mathcal{P}_{i,j},
f_e\leq c_e, \forall e\in P,
\\
-\infty,
&
\text{otherwise}.
\end{cases}
\end{aligned}
\end{equation}

The utility function \eqref{eq: utility} captures two forms of congestion. The first is \emph{soft congestion}, represented by the latency functions, which increase continuously with element loads. The second is \emph{hard congestion}, represented by capacity constraints that cannot be violated. Depending on the application, hard congestion may represent limited infrastructure resources such as vertiport landing slots, electric vehicle charging ports, parking spaces, or airport runway capacity.

\begin{remark}
The value $-\infty$ in \eqref{eq: utility} may equivalently be replaced by a sufficiently negative constant. Since every agent has access to an always-feasible strategy, an infeasible strategy cannot be optimal for any agent and therefore cannot be used in a Nash equilibrium or social optimum.
\end{remark}

\subsubsection{Instances}

We summarize the primitives of the model in the following definition.

\begin{definition}
An \emph{instance} of a nonatomic congestion game with heterogeneous valuations is an eight-tuple $(E,\mathcal{P},r,\ell,c,V,U,D)$, where each component is defined as above.
\end{definition}

\subsection{Relation to Classical Nonatomic Congestion Games}
We show that the nonatomic congestion game model studied in \cite{ROUGHGARDEN2004389} is a special case of our model (under straightforward extensions). To see this, we extend our model in \cref{subsec: val and act} by associating with every strategy $P\in\mathcal{P}$ and element $e\in P$ a consumption rate $\alpha_{P,e}>0$. The element flow is then defined by $f_e
:=
\sum_{P\in\mathcal{P}:e\in P}
\alpha_{P,e}f_P$, and the congestion cost of strategy $P$ by $\ell_P(f)
:=
\sum_{e\in P}
\alpha_{P,e}\ell_e(f_e)$. The model in \cref{subsec: val and act} corresponds to $\alpha_{P,e}=1$.

Consider an instance $(E,\mathcal{P},r,\ell,c,V,U,D)$ of our model with infinite capacities and homogeneous valuations; that is, 
\[
c_e=\infty,
\qquad
D_i=\{1\},
\qquad
V_i=\{0\},
\qquad
\mu_i(\{0\})=r_i,
\]
for every $e\in E$ and $i\in[k]$, and let $U_i(0,\ell)=-\ell$. Then every agent within a group has the same preferences, and an agent choosing strategy $P\in\mathcal{P}_i$ has utility
\[
-\ell_P(f)
=
-\sum_{e\in P}
\alpha_{P,e}\ell_e(f_e).
\]
This recovers the nonatomic congestion-game model of
\cite{ROUGHGARDEN2004389}. 

The singleton valuation space above requires an atomic population measure and therefore falls outside the absolute-continuity assumption used in \cref{subsec: val and act}. It is covered by the general-measures extension in \cref{subsec: abs conti}. In this interpretation, the atom represents a positive mass of infinitesimal agents sharing the same valuation, rather than a single agent with positive mass.

Hence, after allowing strategy-dependent consumption rates and general population measures, classical nonatomic congestion games are a special case of our model. Our model additionally permits agents within the same group to have heterogeneous, multidimensional valuations over multiple outcomes.

\subsection{Nash Equilibria and Social Optima}
\subsubsection{Nash Equilibria}
We adopt the standard notion of Nash equilibrium for nonatomic games \cite{schmeidler1973equilibrium}. Since individual agents have measure zero, equilibrium conditions are imposed almost everywhere. In \cref{app: different equilibrium notions}, we discuss the relationship between this definition and alternative equilibrium notions used in the literature. 

\begin{definition}
\label{def: NE}
An action profile $a$ is a \emph{Nash equilibrium} if, for every $i\in[k]$ and for $\mu_i$-a.e. $v\in V_i$,
$$
u_{i,v}(a_{i,v},a_{-(i,v)})
\geq
u_{i,v}(b,a_{-(i,v)}),
\qquad
\forall b\in\mathcal{P}_i.
$$
\end{definition}

\subsubsection{Social Optima}
Like Nash equilibria, social welfare and social optima are invariant to changes in an action profile on sets of measure zero. 
\begin{definition}
Given an action profile $a$, its \emph{social welfare} is
$$
\begin{aligned}
SW(a)
&:=
\int_V u_{i,v}(a_{i,v})\,d\mu((i,v))
\\
&=
\sum_{i\in[k]}
\int_{V_i}u_{i,v}(a_{i,v})\,d\mu_i(v).
\end{aligned}
$$
\end{definition}

\begin{definition}
\label{def: so}
An action profile $a^*$ is a \emph{social optimum} if it maximizes social welfare; that is,
\begin{equation}
\label{eq: SW}
a^*\in\arg\max_a SW(a).
\end{equation}
\end{definition}

With a slight abuse of terminology, we sometimes refer to the flow induced by a Nash equilibrium or social optimum simply as a Nash equilibrium or social optimum, respectively.

\subsection{Problem Formulation}
An action profile specifies the strategy of every agent and is therefore an infinite-dimensional object. Our main question is whether Nash equilibria and social optima can nevertheless be characterized using finitely many variables. We show in \cref{sec: Characterizations} that Nash equilibria can be characterized by aggregate flows and threshold vectors, while social optima can be characterized using aggregate flows and Kantorovich dual potentials. We also study how to induce a social optimum in routing games in \cref{sec: separable}.

\section{Applications}
\label{sec: app}

In this section, we illustrate several applications of our model. We first show how it extends classical routing games with variable demand \cite{10.5555/1109557.1109630} to allow heterogeneous valuations over multiple destinations. We then discuss applications to market-sharing games \cite{1626428} and congestible club goods \cite{81322e7f-62c2-3874-8382-6c3096359c71,2d696ccb-ad6d-384c-9b47-2eda4c556617,cornes1996theory}.

\subsection{Routing Games with Heterogeneous Valuations}
\label{sec: examples and applications}

\subsubsection{Model}
Consider a directed network $G=(N,E)$ and $k$ \emph{commodities}. Each commodity $i$ is associated with a source $s_i\in N$ and a finite set of destinations indexed by $D_i$. For each $j\in D_i$, let $t_{i,j}\in N$ denote the corresponding destination vertex. We let $\mathcal{P}_{i,j}$ be the nonempty set of simple paths from $s_i$ to $t_{i,j}$ and define $\mathcal{P}_i:=\bigcup_{j\in D_i}\mathcal{P}_{i,j}$.

The rate $r_i>0$ represents the total mass of agents in commodity $i$, and the valuation space $V_i$ describes their valuations over the destinations in $D_i$.

\begin{remark}
We follow the terminology of classical routing games \cite{10.1145/509907.509971,10.5555/1076293} and refer to each group of agents as a \emph{commodity}. The agents themselves may represent travelers, vehicles, or algorithms rather than physical commodities being transported.
\end{remark}

Each edge $e\in E$ has a load-dependent latency function $\ell_e$. We assume that $\ell_e$ is nonnegative, nondecreasing, differentiable, and standard.\footnote{Following Roughgarden \cite{10.1145/509907.509971}, a latency function $\ell$ is standard if $x\ell(x)$ is convex on $\mathbb{R}_+$.} We define its \emph{marginal-cost function} by
\[
\hat{\ell}_e(x)
:=
\left.\frac{d}{dy}\bigl(y\ell_e(y)\bigr)\right|_{y=x}.
\]
We consider separable utilities of the form
\[
U_i(v_j,\ell_P(f))
=
v_j-\ell_P(f).
\]

Classical routing games with variable demand arise as the special case in which $D_i$ consists of group $i$'s destination together with an outside option of not traveling, and every agent's valuation for the outside option is normalized to zero. In that setting, heterogeneity is effectively one-dimensional: an agent's valuation for the real destination alone determines whether she participates. Our model allows agents to hold independently heterogeneous valuations across more outcomes, so that no natural ordering separates participating from non-participating types.

\subsubsection{A Two-Destination Example}
We next consider the simplest setting in which agents choose between multiple destinations. The example illustrates how heterogeneous destination valuations lead to a partition of a multidimensional valuation space.

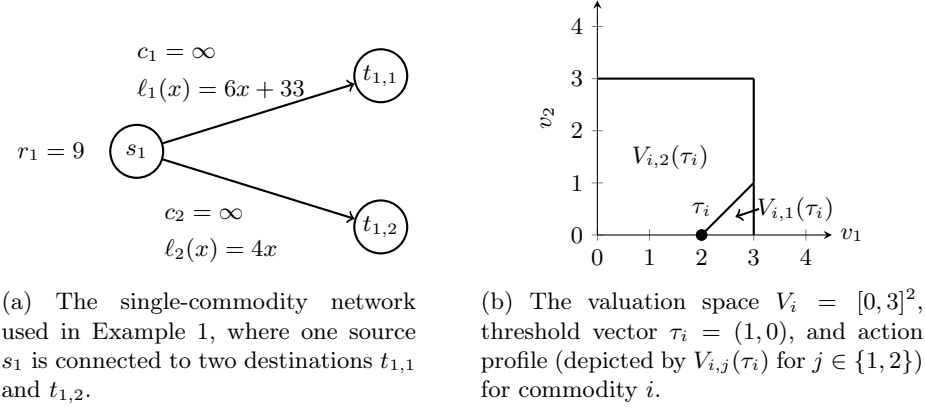
\begin{figure}[]
    \centering
    \begin{subfigure}[t]{0.45\textwidth}
        \centering
        \begin{tikzpicture}[->, thick, node distance=2.5cm]
            \node[circle, draw, minimum size=0.7cm, inner sep=1pt] (A) {$s_1$};
            \node[circle, draw, minimum size=0.7cm, inner sep=1pt, right=of A, yshift=1cm] (B) {$t_{1,1}$};
            \node[circle, draw, minimum size=0.7cm, inner sep=1pt, right=of A, yshift=-1cm] (C) {$t_{1,2}$};

            \node[left=2mm of A] {$r_1=9$};

            \draw (A) -- node[midway, above, xshift=-.5cm] {$\begin{aligned}
                &c_1=\infty\\
                &\ell_1(x)=6x+33
            \end{aligned}$} (B);

            \draw (A) -- node[midway, below, xshift=-.5cm] {$\begin{aligned}
                &c_2=\infty\\
                &\ell_2(x)=4x
            \end{aligned}$} (C);
        \end{tikzpicture}
        \caption{A single-commodity network with one source and two destinations.}
        \label{fig:two-link}
    \end{subfigure}
    \hfill
    \begin{subfigure}[t]{0.45\textwidth}
        \centering
        \begin{tikzpicture}
        \begin{axis}[
            width=0.7\linewidth,
            height=0.7\linewidth,
            axis lines=left,
            xlabel=\(v_1\),
            ylabel=\(v_2\),
            xlabel style={at={(ticklabel* cs:1)}, anchor=west},
            xmin=0, xmax=4.5,
            ymin=0, ymax=4.5,
            xtick={0,1,2,3,4},
            ytick={0,1,2,3,4}
        ]
        \addplot[black, thick] coordinates {(0,3) (3,3)};
        \addplot[black, thick] coordinates {(3,0) (3,3)};
        \addplot[black, thick] coordinates {(2,0) (3,1)};
        \addplot[only marks, mark=*] coordinates {(2,0)};

        \node at (axis cs:2,0.2) [anchor=south] {$\tau_1$};
        \node at (axis cs:3.8,0.5) {$V_{1,1}(\tau_1)$};
        \draw[->][black, thick] (axis cs:3.1,0.5) -- (axis cs:2.65,0.35);
        \node at (axis cs:1.4,1.5) {$V_{1,2}(\tau_1)$};
        \end{axis}
        \end{tikzpicture}
        \caption{The valuation space $V_1=[0,3]^2$ and the Nash equilibrium partition induced by $\tau_1=(2,0)$.}
        \label{fig: threshold}
    \end{subfigure}
    \caption{A two-destination routing game and the induced partition of the valuation space.}
\end{figure}

\begin{example}
\label{ex: two 1}
Consider the unconstrained ($c=\infty$) single-commodity network in \cref{fig:two-link}. The source $s_1$ is connected to two destinations indexed by $D_1=\{1,2\}$, corresponding to $t_{1,1}$ and $t_{1,2}$. The total rate is $r_1=9$, and the valuation space $V_1=[0,3]^2$ is endowed with the Lebesgue measure.

At the Nash equilibrium, the valuation space is partitioned according to the threshold vector $\tau_1=(2,0)$. Specifically,
\begin{align*}
V_{1,1}(\tau_1)
&=
\{v\in V_1\mid v_1-2\geq v_2\},\\
V_{1,2}(\tau_1)
&=
\{v\in V_1\mid v_1-2<v_2\}.
\end{align*}
Agents in $V_{1,1}(\tau_1)$ travel to $t_{1,1}$, while agents in $V_{1,2}(\tau_1)$ travel to $t_{1,2}$. The resulting flows are
\[
f_1=0.5,
\qquad
f_2=8.5.
\]
Hence, $\ell_1(f_1)=36$ and $\ell_2(f_2)=34$. An agent therefore prefers the first destination precisely when $v_1-36\geq v_2-34$, or equivalently, when $v_1-2\geq v_2$. Thus, the above partition constitutes a Nash equilibrium. We formally define these threshold regions and characterize Nash equilibria in \cref{subsec: characterization of nash}.

The social optimum has a similar structure but a different threshold. The marginal-cost functions are
\[
\hat{\ell}_1(x)=12x+33,
\qquad
\hat{\ell}_2(x)=8x.
\]
The social optimum is induced by $\tau_1=(1,0)$, which yields
\[
f_1=2,
\qquad
f_2=7.
\]
The corresponding marginal costs are $\hat{\ell}_1(f_1)=57$ and $\hat{\ell}_2(f_2)=56$. Hence, an agent is assigned to the first destination precisely when $v_1-57\geq v_2-56$, or equivalently, when $v_1-1\geq v_2$.

This example cannot be reduced to a classical single-destination routing game. One might introduce a common destination $T$ and connect both $t_{1,1}$ and $t_{1,2}$ to $T$ through zero-latency edges. In the resulting single-destination game, however, all used paths between the same source and destination must have equal latency at equilibrium. This is not true in the original game: the two destinations have equilibrium latencies $36$ and $34$. More fundamentally, destination choice is determined by a partition of the two-dimensional valuation space. This structure does not arise in the classical single-destination model. We characterize such partitions in \cref{sec: Characterizations}.
\end{example}

\subsection{Market-Sharing Games}

Our model can also represent a nonatomic version of the market-sharing games studied by Goemans et al. \cite{1626428}. Consider a finite set of markets $D=\{1,\ldots,m\}$. Each agent has a valuation vector $v=(v_1,\ldots,v_m)$, where $v_j$ represents its value for serving market $j$.

Let $n_j$ denote the size of market $j$ and $f_j$ the mass of agents serving that market. If the market is shared equally among its participants, an agent serving market $j$ receives a share $n_j/f_j$. Its utility is therefore $v_j \cdot n_j / f_j$.
Defining $\ell_j(f)= f_j / n_j$, we can write this utility as $U_i(v_j,\ell_j(f))
= v_j/\ell_j(f)$. Thus, the nonatomic market-sharing game is an instance of our model with the nonseparable utility function $U_i(x, y) = x/y$.\footnote{Strictly speaking, the expression $x/y$ is
singular at $y=0$ and therefore does not satisfy our continuity
assumption. A regularized version fits our model exactly: for any
$\varepsilon_0,\varepsilon_1>0$, let
\[
U_i(x,y)=\frac{x+\varepsilon_0}{y+\varepsilon_1},
\qquad
\ell_j(f)=\frac{f_j}{n_j}.
\]
Then
\[
U_i(v_j,\ell_j(f))
=
\frac{(v_j+\varepsilon_0)n_j}
{f_j+\varepsilon_1 n_j},
\]
and $U_i$ is continuous, strictly increasing in its first argument,
and strictly decreasing in its second argument.}


\subsection{Congestible Club Goods}
Our model also applies to congestible club goods, introduced by Buchanan \cite{81322e7f-62c2-3874-8382-6c3096359c71}; see also \cite{2d696ccb-ad6d-384c-9b47-2eda4c556617,cornes1996theory}. Examples include coworking spaces, fitness centers, golf clubs, and other membership-based facilities that provide excludable but congestible services.

Let $D=\{1,\ldots,m\}$ denote a finite set of clubs. Each agent chooses a club to join and has a valuation vector $v=(v_1,\ldots,v_m)$, where $v_j$ represents the agent's valuation for membership in club $j$. If a mass $f_j$ of agents joins club $j$, members experience congestion represented by $\ell_j(f_j)$. Their utility can therefore be written as $U_i(v_j,\ell_j(f_j))$.

Capacity constraints can represent membership limits or physical capacity. More general constraints can represent resources shared across multiple clubs; see \cref{sec: convex}.


\section{Characterizations of Nash Equilibria and Social Optima}
\label{sec: Characterizations}
In this section, we show that, although Nash equilibria and social optima are defined in terms of infinite-dimensional action profiles, both admit finite-dimensional characterizations, as illustrated in \cref{ex: two 1}. The two characterizations, however, have different structures.

For a Nash equilibrium, it suffices to specify the aggregate flow and one threshold vector for each group. For a social optimum, we must additionally characterize how agents are assigned across strategies associated with the same outcome.

The main difficulty comes from multidimensional heterogeneous valuations. In classical congestion games with homogeneous agents, the aggregate flow is sufficient to describe equilibria and social optima \cite{ROUGHGARDEN2004389}. With heterogeneous valuations, the aggregate flow does not specify which types choose which outcomes.

In routing games with elastic demand, each commodity has a single destination. The allocation can therefore be described by a threshold on a one-dimensional valuation space \cite{10.5555/1109557.1109630}. With multiple outcomes, valuations are multidimensional and there is no natural ordering of types. We show that a finite-dimensional threshold structure nevertheless persists.

\subsection{Characterization of Nash Equilibria}
\label{subsec: characterization of nash}

\subsubsection{Existence and Nonexistence}

We first consider the existence of Nash equilibria. Hard capacity constraints introduce discontinuities into agents' utilities. As a result, a Nash equilibrium need not exist in the general capacitated setting. We provide a counterexample in \cref{subsec: ex 3}.

When capacities are nonbinding ($c=\infty$), these discontinuities disappear. Existence then follows from Schmeidler's theorem for nonatomic games \cite{schmeidler1973equilibrium}.

\begin{proposition}
\label{prop: ne exists}
Given an instance $(E,\mathcal{P},r,\ell,c,V,U,D)$ with $c=\infty$, a Nash equilibrium exists.
\end{proposition}

\begin{proof}
We apply the standard extension of \cite[Theorem 2]{schmeidler1973equilibrium} to finite nonatomic measure spaces. The player space $(V,\mathcal{B}_V,\mu)$ is nonatomic, and the action space $\mathcal{P}$ is finite. We verify the conditions of the theorem.

First, for every agent and strategy, utility is continuous in the aggregate strategy distribution. Each edge latency $\ell_e$ is continuous, so every strategy latency $\ell_P(f)$ is continuous in the aggregate flow $f$. Since $U_i$ is continuous, the resulting utility is also continuous in $f$.

Second, the relevant preference sets are measurable. For any $P_1\in\mathcal{P}_{i,j_1}$ and $P_2\in\mathcal{P}_{i,j_2}$, the set
\[
\left\{
v\in V_i
\;\middle|\;
U_i(v_{j_1},\ell_{P_1}(f))
>
U_i(v_{j_2},\ell_{P_2}(f))
\right\}
\]
is Borel measurable because $U_i$ is continuous.

Finally, an agent's utility depends on the action profile only through her own strategy and the induced aggregate flow. The aggregate flow is determined by the aggregate distribution of strategies. Thus, the additional condition of \cite[Theorem 2]{schmeidler1973equilibrium} is satisfied.

Therefore, there exists a pure-strategy Nash equilibrium.
\end{proof}

With hard capacities, the continuity condition used above no longer holds. In particular, an individual agent has measure zero and therefore cannot change the aggregate flow through a unilateral deviation. At a binding capacity, this can lead to nonexistence, as shown in \cref{subsec: ex 3}.

\subsubsection{Threshold-Based Characterization}

We next characterize a Nash equilibrium whenever one exists. For a threshold vector $\tau_i\in\mathbb{R}_+^{|D_i|}$, define
\begin{equation}
\label{eq: threshold}
V_{i,j}(\tau_i)
:=
\left\{
v\in V_i
\;\middle|\;
j\in
\arg\max_{\hat j\in D_i}
U_i(v_{\hat j},\tau_{i,\hat j})
\right\}.
\end{equation}
These sets form an a.e.-partition of $V_i$.\footnote{\label{footnote: disjoint}
To see this, consider two distinct outcomes $j_1,j_2\in D_i$. Their indifference set is
\[
S
=
\left\{
v\in V_i
\;\middle|\;
U_i(v_{j_1},\tau_{i,j_1})
=
U_i(v_{j_2},\tau_{i,j_2})
\right\}.
\]
Fix all coordinates of $v$ except $v_{j_2}$. Since $U_i$ is strictly increasing in its first argument, there is at most one value of $v_{j_2}$ satisfying the equality. Hence, $S$ has Lebesgue measure zero. Since $\mu_i$ is absolutely continuous with respect to the Lebesgue measure, $S$ is also $\mu_i$-null. Since $D_i$ is finite, the union of all pairwise indifference sets is $\mu_i$-null.}
Some of the sets $V_{i,j}(\tau_i)$ may be empty. An example is illustrated in \cref{fig: threshold}.

The threshold vector determines how agents choose among outcomes. An agent with valuation $v$ chooses an outcome $j$ that maximizes $U_i(v_j,\tau_{i,j})$. This motivates the following definition.

\begin{definition}
An action profile $a$ is \emph{threshold-based} if, for every group $i$, there exists a threshold vector $\tau_i\in\mathbb{R}_+^{|D_i|}$ such that, for $\mu_i$-a.e. $v\in V_i$, we have $a_{i,v}\in\mathcal{P}_{i,j}
\text{ if and only if }
v\in V_{i,j}(\tau_i).$
\end{definition}

We can now characterize Nash equilibria using only aggregate flows and threshold vectors.

\begin{theorem}
\label{prop: ne}
Given an instance $(E,\mathcal{P},r,\ell,c,V,U,D)$, every Nash equilibrium is threshold-based. Moreover, an admissible flow $f$ is induced by a Nash equilibrium if and only if the following conditions hold:
\begin{enumerate}
    \item $f_e\leq c_e$ for every $ e\in E$.

    \item For every $i\in[k]$, $j\in D_i$, and $P,Q\in\mathcal{P}_{i,j}$ with $f_P>0$,
    \[
    \ell_P(f)\leq\ell_Q(f).
    \]

    \item For every $i\in[k]$ and $j\in D_i$,
    \[
    \sum_{P\in\mathcal{P}_{i,j}}f_P
    =
    \mu_i\bigl(V_{i,j}(\tau_i)\bigr),
    \]
    where $\tau_{i,j}
    :=
    \min_{P\in\mathcal{P}_{i,j}}\ell_P(f)$.
\end{enumerate}
\end{theorem}

\begin{proof}
Let $a$ be a Nash equilibrium and let $f$ be its induced flow. For every group $i$ and outcome $j$, define
\[
\tau_{i,j}
:=
\min_{P\in\mathcal{P}_{i,j}}\ell_P(f).
\]

First, $f$ must satisfy all capacity constraints. Otherwise, some agents would receive utility $-\infty$ and could deviate to an always-feasible strategy.

Next, consider $P,Q\in\mathcal{P}_{i,j}$ with $f_P>0$. If $\ell_P(f)>\ell_Q(f)$, then every agent using $P$ could switch to $Q$. Since both strategies yield the same outcome $j$ and $U_i$ is strictly decreasing in its second argument, such a deviation would strictly increase utility. Therefore, $\ell_P(f)\leq\ell_Q(f)$. Hence, every strategy with positive flow in $\mathcal{P}_{i,j}$ has latency $\tau_{i,j}$.

It follows that an agent with valuation $v$ who chooses outcome $j$ receives $U_i(v_j,\tau_{i,j})$. At equilibrium, the chosen outcome must maximize this quantity. Therefore, for $\mu_i$-a.e. $v$, $a_{i,v}\in\mathcal{P}_{i,j}$ if and only if  $v\in V_{i,j}(\tau_i)$. Thus, the equilibrium is threshold-based, and $\sum_{P\in\mathcal{P}_{i,j}}f_P
=
\mu_i\bigl(V_{i,j}(\tau_i)\bigr)$.

Conversely, suppose that an admissible flow $f$ satisfies the three conditions. Fix a group $i$ and an outcome $j$. By the third condition, $\sum_{P\in\mathcal{P}_{i,j}}f_P
=
\mu_i\bigl(V_{i,j}(\tau_i)\bigr)$. Since $\mu_i$ is nonatomic, there exists a measurable partition
$\{A_P\}_{P\in\mathcal{P}_{i,j}}$ of $V_{i,j}(\tau_i)$ such that $\mu_i(A_P)=f_P$ for any $P\in\mathcal{P}_{i,j}$ (cf. \cite[Theorem 10.52]{aliprantis2006infinite}).

Assign every $v\in A_P$ to strategy $P$. Repeating this construction for every group and outcome gives a measurable action profile $a$ that induces $f$.

By the second condition, every strategy with positive flow in $\mathcal{P}_{i,j}$ has minimum latency $\tau_{i,j}$. By the definition of $V_{i,j}(\tau_i)$, agents assigned to outcome $j$ weakly prefer that outcome to every other outcome. Thus, no agent can obtain a strictly higher utility by deviating to another strategy. Therefore, $a$ is a Nash equilibrium.
\end{proof}

The three conditions have simple interpretations. The first requires the aggregate flow to respect the capacity constraints. The second requires agents choosing the same outcome to use only minimum-latency strategies. The third determines how the population is divided among outcomes.

For example, consider two outcomes $j_1,j_2\in D_i$. Their indifference boundary is $U_i(v_{j_1},\tau_{i,j_1})
=
U_i(v_{j_2},\tau_{i,j_2})$. This boundary divides the corresponding valuation coordinates into regions associated with the two outcomes. The two-destination example in \cref{fig: threshold} illustrates this geometry.

\subsubsection{Computational Approach}

\cref{prop: ne} reduces the characterization of an equilibrium from an infinite-dimensional action profile to finitely many aggregate flows and threshold vectors. Equilibrium computation can therefore be formulated as a finite-dimensional fixed-point problem. Establishing convergence guarantees for specific computational methods is left to future work.

\subsection{Characterization of Social Optima}

Similar to \cref{prop: ne}, we seek a finite-dimensional
characterization of a social optimum. We first decompose the social
welfare maximization problem into an upper-level optimization over
aggregate flows and a lower-level assignment problem. For a fixed
aggregate flow, the lower-level problem is a semi-discrete optimal
transport problem. Kantorovich duality then yields a finite-dimensional
characterization in terms of strategy-level dual potentials.

\subsubsection{Upper-Level Problem}

For a fixed aggregate flow, the congestion experienced on every
strategy is fixed. The remaining problem is to assign agents to
strategies while preserving the prescribed strategy masses.

For a fixed admissible flow $f$, let 
\[
\mathcal{T}_i(f)
:=\left\{
T:V_i\to\mathcal{P}_i \text{ measurable}
\;\middle|\; 
T_\# \mu_i = f
\right\},
\]
where $T_\#\mu_i$ denotes the pushforward of $\mu_i$ under $T$; that is, $\mu_i(T^{-1}(P)) = f_P$ for every $P \in \mathcal{P}_i$. Here, $T:V_i\to\mathcal{P}_i$ assigns each agent in group $i$ to a strategy. The constraints require $T$ to induce the prescribed flow $f$ on $\mathcal{P}_i$. With $\mathcal{T}_i(f)$, we define the value function
\begin{equation}
\label{eq: lower 0}
\mathcal{S}_i(f)\!
:=\!\!\!
\sup_{T\in\mathcal{T}_i(f)}
\sum_{j\in D_i}
\sum_{P\in\mathcal{P}_{i,j}}
\int_{T^{-1}(P)}\!\!\!\!\!
U_i(v_j,\ell_P(f))\,d\mu_i(v).
\end{equation}

\begin{proposition}
\label{prop: max sw 0}
Given an instance $(E,\mathcal{P},r,\ell,c,V,U,D)$, the social welfare
optimization problem \eqref{eq: SW} is equivalent to
\begin{subequations}
\label{eq: upper}
\begin{align}
\label{eq: upper obj}
\sup_f \quad
&\sum_{i\in[k]}\mathcal{S}_i(f)
\\
\text{s.t.}\quad
&f_P\geq 0,
\qquad
\forall P\in\mathcal{P},
\label{eq: upper con1}
\\
&\sum_{P\in\mathcal{P}_i}f_P=r_i,
\qquad
\forall i\in[k],
\label{eq: upper con2}
\\
&f_e\leq c_e,
\qquad
\forall e\in E.
\label{eq: upper con3}
\end{align}
\end{subequations}
\end{proposition}

\begin{proof}
Let $a$ be an action profile whose induced flow $f$ is capacity-feasible. Its social welfare is
\[
SW(a)
=
\sum_{i\in[k]}
\sum_{j\in D_i}
\sum_{P\in\mathcal{P}_{i,j}}
\int_{a_i^{-1}(P)}
U_i(v_j,\ell_P(f))\,d\mu_i(v),
\]
where $a_i$ denotes the restriction of $a$ to group $i$.

Fixing $f$ determines the latency of every strategy. Thus, conditional
on $f$, maximizing social welfare reduces to assigning the agents in
each group to strategies while preserving the prescribed strategy
masses. This is precisely the problem defining $\mathcal{S}_i(f)$.

Conversely, any flow satisfying
\eqref{eq: upper con1}--\eqref{eq: upper con3} is admissible and
capacity-feasible. By \cref{lemma: admissible flow}, it can be induced
by an action profile. Optimizing the assignment within each group
therefore gives $\mathcal{S}_i(f)$. Maximizing over all
capacity-feasible flows yields \eqref{eq: upper}.
\end{proof}

\subsubsection{Semi-Discrete Optimal Transport}
Before characterizing the lower-level problem, we briefly introduce the optimal transport concepts needed for our analysis. For a comprehensive treatment, see \cite{villani2008optimal}. The notation in this subsection is local and independent of that used elsewhere in the paper.

Optimal transport studies how to redistribute mass from one measure to another while optimizing a given objective. Let $\mu$ and $\nu$ be finite measures on spaces $X$ and $Y$, respectively, with equal total mass. A \emph{transport plan}, or \emph{coupling}, between $\mu$ and $\nu$ is a measure $\pi$ on $X\times Y$ whose marginals are $\mu$ and $\nu$. We denote the set of such transport plans by $\Pi(\mu,\nu)$. Given a measurable reward function $c:X\times Y\to\mathbb{R}$, the Kantorovich optimal transport problem is
\[
\sup_{\pi\in\Pi(\mu,\nu)}
\int_{X\times Y}c(x,y)\,d\pi(x,y).
\]
Thus, a transport plan specifies how the mass at each point of $X$ is distributed across points of $Y$.

A \emph{transport map} is a measurable function $T:X\to Y$ satisfying $T_{\#}\mu=\nu$, where $T_{\#}\mu$ denotes the pushforward of $\mu$ under $T$. Such a map induces the transport plan $\pi=(\operatorname{Id},T)_{\#}\mu$, where $\operatorname{Id}:X\to X$ denotes the identity map. Transport maps therefore correspond to deterministic assignments,
whereas general transport plans allow the mass at a point $x$ to be split across multiple points of $Y$. The optimization over transport maps is commonly referred to as the Monge problem, while the Kantorovich problem above is its relaxation over transport plans.

An optimal transport problem is called \emph{semi-discrete} when one of the two measures is supported on a finite set. In our setting, the source measure represents a continuum of agents, while the target
measure represents the aggregate masses assigned to a finite set of
strategies. A transport map therefore corresponds directly to a
deterministic assignment of agents to strategies.

The following lemma combines Kantorovich duality with a purification
argument for nonatomic measures. It gives a finite-dimensional dual
formulation and shows that an optimal transport map can be supported
by any optimal dual solution.

\begin{lemma}
\label{lemma: transport map}
Let $X$ be a Polish space\footnote{A Polish space is a  separable, completely metrizable topological space.} and let $\mu$ be a nonatomic finite Borel measure on $X$. Let $(Y,2^Y,\nu)$ be a finite
measure space with $Y=\{1,\ldots,m\}$ and $\nu(Y)=\mu(X)$. Let $c:X\times Y\to\mathbb{R}$ be bounded and continuous.

For each $j\in Y$, let $\nu_j:=\nu(\{j\})$. Then
\begin{equation}
\label{eq: semi discrete duality}
\begin{aligned}
&\sup_{\pi\in\Pi(\mu,\nu)}
\int_{X\times Y}c(x,j)\,d\pi(x,j)\\
&=
\min_{q\in\mathbb{R}^m}
\left\{
\int_X
\max_{j\in Y}
\bigl(c(x,j)-q_j\bigr)\,d\mu(x)
+
\sum_{j=1}^m q_j\nu_j
\right\},
\end{aligned}
\end{equation}
and both optima are attained.

Moreover, for every optimal dual solution $q$, there exists a
measurable map $T:X\to Y$ such that $T_{\#}\mu=\nu$ and
\[
T(x)\in
\arg\max_{j\in Y}
\bigl(c(x,j)-q_j\bigr)
\quad
\text{for $\mu$-a.e. }x.
\]
The transport plan induced by $T$ is optimal.
\end{lemma}
\begin{proof}
See \cref{app: proof}.
\end{proof}

\subsubsection{Lower-Level Problem}

We next characterize the lower-level assignment problem
\eqref{eq: lower 0}. We show that an optimal assignment can be
described by finitely many dual potentials together with a measurable
partition of the valuation space.

For $q\in\mathbb{R}^{|\mathcal{P}_i|}$, define
\[
V_{i,P}(q)
:=
\left\{
v\in V_i
\;\middle|\;
P\in
\arg\max_{\hat P\in\mathcal{P}_i}
\left[
U_i(v_{\hat j},\ell_{\hat P}(f))
-q_{\hat P}
\right]
\right\},
\]
where $\hat j$ denotes the outcome associated with $\hat P$.

Unlike the equilibrium regions $V_{i,j}(\tau_i)$ in
\cref{prop: ne}, the sets $V_{i,P}(q)$ need not be a.e.-disjoint.
For example, if $P,Q\in\mathcal{P}_{i,j}$ have identical latencies
and dual potentials, then
\[
U_i(v_j,\ell_P(f))-q_P
=
U_i(v_j,\ell_Q(f))-q_Q
\]
for every $v\in V_i$.

The proof approach is as follows. We identify \eqref{eq: lower 0} as a Monge optimal transport problem,
in which transport plans are restricted to those induced by
measurable maps. We then consider its Kantorovich relaxation over
general transport plans. Applying \cref{lemma: transport map} shows
that the relaxation admits an optimal transport map and is therefore
exact. This also yields a finite-dimensional dual characterization of
the lower-level problem.

\begin{theorem}
\label{prop: kantorovich 0}
Fix an admissible flow $f$. A solution to \eqref{eq: lower 0} exists, and $\mathcal{S}_i(f)
=
\min_{w\in\mathbb{R}^{\mathcal{P}_i}}
\Phi_i(w,f)$, where
\begin{equation}
\label{eq: dual phi 0}
\begin{aligned}
\Phi_i(w,f)
:={}&
\int_{V_i}
\max_{\hat P\in\mathcal{P}_i}
\left[
U_i(v_{\hat j},\ell_{\hat P}(f))-w_{\hat P}
\right]
d\mu_i(v)\\
&+
\sum_{P\in\mathcal{P}_i}w_Pf_P.
\end{aligned}
\end{equation}

An optimal dual vector $q$ can be chosen such that $q\in [0, 2K]^{|\mathcal{P}_i|}$ for some $K>0$ independent of $f$. For every optimal $q$ satisfying
this bound, there exists an optimal transport map
$T:V_i\to\mathcal{P}_i$ such that $T_{\#}\mu_i=f^i$ and 
\[
T(v)
\in
\arg\max_{\hat P\in\mathcal{P}_i}
\left[
U_i(v_{\hat j},\ell_{\hat P}(f))-q_{\hat P}
\right]
\]
for $\mu_i$-a.e. $v$, or equivalently, $T^{-1}(P)\subseteq V_{i,P}(q)$ up to a $\mu_i$-null set for every $P\in\mathcal{P}_i$.
\end{theorem}

\begin{proof}
For each $P\in\mathcal{P}_{i,j}$, define $c_f(v,P)
:=
U_i(v_j,\ell_P(f))$. Since $f$ is admissible,
\[
f^i(\mathcal{P}_i)
=
\sum_{P\in\mathcal{P}_i}f_P
=
r_i
=
\mu_i(V_i).
\]
Thus, \eqref{eq: lower 0} is a Monge optimal transport problem from $\mu_i$ to $f^i$ with surplus $c_f$. 

Its Kantorovich relaxation is
\begin{equation}
\label{eq: kant primal}
\sup_{\pi\in\Pi(\mu_i,f^i)}
\int_{V_i\times\mathcal{P}_i}
c_f(v,P)\,d\pi(v,P).
\end{equation}
Applying \cref{lemma: transport map} (with $X=V_i,\
Y=\mathcal{P}_i,\
\mu=\mu_i,\
\nu=f^i,\
c=c_f$) gives
\[
\sup_{\pi\in\Pi(\mu_i,f^i)}
\int_{V_i\times\mathcal{P}_i}
c_f(v,P)\,d\pi(v,P)
=
\min_{w\in\mathbb{R}^{\mathcal{P}_i}}
\Phi_i(w,f).
\]
Moreover, for every optimal dual vector $q$, there exists a measurable
map $T:V_i\to\mathcal{P}_i$ such that
\[
T(v)
\in
\arg\max_{\hat P\in\mathcal{P}_i}
\left[
U_i(v_{\hat j},\ell_{\hat P}(f))
-q_{\hat P}
\right]
\]
for $\mu_i$-a.e. $v$, and $\mu_i(T^{-1}(P))=f_P$ for every $P\in\mathcal{P}_i$. 
Hence, $T\in\mathcal{T}_i(f)$. By
\cref{lemma: transport map}, the transport plan induced by $T$ is
optimal for \eqref{eq: kant primal}. Therefore, the Kantorovich
relaxation is exact, $T$ solves \eqref{eq: lower 0}, and $\mathcal{S}_i(f)
=
\min_{w\in\mathbb{R}^{\mathcal{P}_i}}
\Phi_i(w,f)$. The inclusion $T^{-1}(P)\subseteq V_{i,P}(q)$ up to a $\mu_i$-null set follows directly from the definition of
$V_{i,P}(q)$.

It remains to establish the uniform bound on the dual potentials.
Since the set of admissible flows is compact, $V_i$ is compact, and
the latency and utility functions are continuous, there exists
$K>0$, independent of $f$, such that
\[
\left|U_i(v_j,\ell_P(f))\right|\leq K
\]
for every admissible flow $f$, every $v\in V_i$, and every
$P\in\mathcal{P}_{i,j}$.

For any $a\in\mathbb{R}$, since $\sum_{P\in\mathcal{P}_i}f_P
=
\mu_i(V_i)$, we have
\[
\Phi_i(w+a\mathbf{1},f)
=
\Phi_i(w,f),
\]
where $\mathbf{1}$ denotes the all-ones vector.

Thus, any optimal dual vector can be normalized so that $\min_{P\in\mathcal{P}_i}w_P=0$. Suppose $w_P>2K$ for some $P$. Then
\[
U_i(v_j,\ell_P(f))-w_P<-K
\]
for every $v\in V_i$. Since some coordinate of $w$ equals zero, there
exists $\hat P\in\mathcal{P}_i$ such that
\[
U_i(v_{\hat j},\ell_{\hat P}(f))-w_{\hat P}\geq -K.
\]
Hence, $P$ never attains the pointwise maximum in
\eqref{eq: dual phi 0}. Replacing $w_P$ by $2K$ therefore leaves the
integral term unchanged and weakly decreases the linear term.
Repeating this argument for every coordinate exceeding $2K$ yields an
optimal dual vector $q\in[0,2K]^{|\mathcal{P}_i|}$.
\end{proof}

\begin{remark}
Although both $\tau$ and $q$ determine regions of the valuation space,
they have different origins. The vector $\tau$ parametrizes the
equilibrium threshold regions, whereas $q$ is a Kantorovich dual
potential supporting an optimal assignment.
\end{remark}

\begin{remark}
The dual potential $q_P$ can be viewed as a strategy-specific
supporting price for the lower-level assignment problem conditional
on a fixed aggregate flow. It should not, in general, be interpreted
as an actual congestion price. In the separable-utility setting considered
in \cref{sec: separable}, the dual potentials admit a more explicit
interpretation through the first-order conditions of the social
planner's problem.
\end{remark}

\subsubsection{Existence and Comparative Statics}
\label{subsubsec: existence}
For an admissible flow $f$, let $\Lambda_i(f)$ be the set of optimal dual potentials:
\[
\Lambda_i(f)
:=
\arg\min_{q\in[0,2K]^{|\mathcal{P}_i|}}
\Phi_i(q,f).
\]

The next result establishes basic regularity properties of the
lower-level problem and, together with the upper-level formulation, the existence of a social optimum. We also show that the set of optimal dual potentials has a lattice structure. In particular, $\Lambda_i(f)$ admits a smallest and a largest element under the product order. This provides canonical extremal selections of the dual potentials and is useful for studying their comparative statics \cite{topkis1998supermodularity}. We refer the reader to \cref{app: lattice} for the relevant definitions and results from lattice theory.

\begin{corollary}
\label{cor: exist opt}
Given an instance $(E,\mathcal{P},r,\ell,c,V,U,D)$, a social optimum
exists. Moreover, for every group $i$:
\begin{enumerate}
    \item The value function $f\mapsto\mathcal{S}_i(f)$ is continuous, where $\mathcal{S}_i(f)$ is defined in \cref{prop: kantorovich 0}.

    \item The correspondence $f\mapsto\Lambda_i(f)$ is upper hemicontinuous with nonempty compact values.

    \item For every admissible flow $f$, the set
    $\Lambda_i(f)$ is a complete lattice under the product order.
\end{enumerate}
\end{corollary}

\begin{proof}
Since $U_i$, the latency functions, and the maximum over the finite set $\mathcal{P}_i$ are continuous, $\Phi_i(q,f)$ is jointly continuous in $(q,f)$. The feasible set $[0,2K]^{|\mathcal{P}_i|}$ is compact and independent of $f$. Berge's maximum theorem
(cf. \cite[Theorem 17.31]{aliprantis2006infinite}) therefore implies that
$\mathcal{S}_i(f)$ is continuous and that $\Lambda_i(f)$ is upper
hemicontinuous with nonempty compact values.

The set of capacity-feasible flows is compact. By
\cref{prop: max sw 0}, the social welfare problem has the
finite-dimensional representation $\max_f\sum_{i\in[k]}\mathcal{S}_i(f)$ over this compact set. Since the objective is continuous, a social
optimum exists.

It remains to establish the lattice property. Fix $f$ and $v$. Define $h(q)
:=
\max_{P\in\mathcal{P}_i}
\{a_P-q_P\}$ and $a_P:=U_i(v_j,\ell_P(f))$, where $j$ is such that $P\in\mathcal{P}_{i,j}$.

We claim that $h$ is submodular. Let $q^1,q^2\in\mathbb{R}^{|\mathcal P_i|}$,
and choose
\[
P^*
\in
\arg\max_{P\in\mathcal P_i}
\{a_P-(q^1\wedge q^2)_P\}.
\]
Without loss of generality, suppose $(q^1\wedge q^2)_{P^*}=q^1_{P^*}$. Then
\[
h(q^1\wedge q^2)
=
a_{P^*}-q^1_{P^*}
\leq
h(q^1).
\]
Moreover, since $q^1\vee q^2\geq q^2$ componentwise and $h$ is
componentwise decreasing, $h(q^1\vee q^2)\leq h(q^2)$. Therefore,
\[
h(q^1\vee q^2)+h(q^1\wedge q^2)
\leq
h(q^1)+h(q^2).
\]
Hence, $h$ is submodular.

Integration preserves submodularity, and the term $\sum_{P\in\mathcal P_i}q_Pf_P$ is modular. Thus, $\Phi_i(\cdot,f)$ is submodular. By
\cite[Theorem 2.7.1]{topkis1998supermodularity},
$\Lambda_i(f)$ is a sublattice of
$[0,2K]^{|\mathcal{P}_i|}$.

Finally, $\Lambda_i(f)$ is compact. By
\cite[Theorem 2.3.1]{topkis1998supermodularity}, every compact
sublattice of $\mathbb{R}^{|\mathcal{P}_i|}$ is complete. Hence,
$\Lambda_i(f)$ is a complete lattice.
\end{proof}

\subsubsection{Computational Approach}

\cref{prop: kantorovich 0} reduces the lower-level assignment problem,
for a fixed aggregate flow $f$, to the finite-dimensional convex
optimization problem $\min_q\Phi_i(q,f)$. By \cref{lemma: transport map}, every optimal dual potential supports an optimal transport map. 

The remaining optimization is over the finite-dimensional aggregate
flow $f$. In the general model, this upper-level problem need not be
convex. Nevertheless, the original infinite-dimensional social welfare
problem is reduced to a finite-dimensional optimization over aggregate
flows, with a finite-dimensional convex optimal transport problem at
each candidate flow.

In the separable-utility setting considered in
\cref{sec: separable}, the upper-level problem has additional convex structure, which permits a more direct characterization of the social optimum.

\section{Inducing Social Optima under Separable Utilities}
\label{sec: separable}

The characterization in \cref{sec: Characterizations} applies to
arbitrary utility functions. We now consider separable utilities of
the form
\[
U_i(v_j,\ell_P(f))=v_j-\ell_P(f).
\]
The routing-game model introduced in
\cref{sec: examples and applications} is an important special case.

Under separable utilities, the valuation and congestion terms can be
separated. As a result, the lower-level assignment problem assigns
agents only to outcomes, so its dual potentials are indexed by
outcomes rather than strategies. Moreover, under the assumptions on
the element latency functions introduced below, the social-welfare
optimization problem becomes a convex optimization problem in the
aggregate flows. This structure yields a first-order characterization
of social optima together with an explicit pricing interpretation.

Throughout this section, we assume that each $\ell_e$ is differentiable and standard (see \cref{sec: examples and applications}) and that constraint qualifications for the KKT conditions hold.

\subsection{Social Welfare Optimization Problem}
\label{app: sw opt}

Given an admissible flow $f$ and a commodity $i$, let $\nu_{f,i}$ be
the measure on $D_i$ induced by $f$, and write
\[
\nu_{f,i}(j)
:=
\nu_{f,i}(\{j\})
=
\sum_{P\in\mathcal{P}_{i,j}}f_P,
\qquad
j\in D_i.
\]
Thus, $\nu_{f,i}(j)$ is the total mass of commodity $i$ assigned to outcome $j$.

The next two propositions specialize
\cref{prop: max sw 0,prop: kantorovich 0} to separable utilities.
Since the valuation and latency terms are separated, the lower-level
problem takes a simpler form. Throughout this section, $\mathcal S_i(f)$ denotes the optimal
value of the valuation-only assignment problem defined below, which excludes the latency costs included in \eqref{eq: lower 0}.

\begin{proposition}
\label{prop: max sw}
Given an instance $(E,\mathcal P,r,\ell,c,V,U,D)$ with separable utilities $U_i(x,y)=x-y$, the social welfare optimization
problem \eqref{eq: SW} is equivalent to
\begin{subequations}
\label{eq: upper 0}
\begin{align}
\label{eq: upper obj 0}
\sup_f \quad
&\sum_{i\in[k]}\mathcal{S}_i(f)
-\sum_{P\in\mathcal{P}}f_P\ell_P(f)
\\
\label{eq: upper con1 0}
\text{s.t.}\quad
&f_P\geq 0,
&&\forall P\in\mathcal{P},
\\
\label{eq: upper con2 0}
&\sum_{P\in\mathcal{P}_i}f_P=r_i,
&&\forall i\in[k],
\\
\label{eq: upper con3 0}
&f_e\leq c_e,
&&\forall e\in E,
\end{align}
\end{subequations}
where
\begin{equation}
\label{eq: lower}
\mathcal{S}_i(f)
:=
\sup_{\substack{T:V_i\to D_i\text{ measurable}\\
T_{\#}\mu_i=\nu_{f,i}}}
\sum_{j\in D_i}
\int_{T^{-1}(j)}
v_j\,d\mu_i(v).
\end{equation}
\end{proposition}
\begin{proof}
Since
\[
U_i(v_j,\ell_P(f))
=
v_j-\ell_P(f),
\]
the latency term depends only on the aggregate flow and is independent
of the assignment of agents to outcomes. Hence, for an action
profile that optimally induces $f$,
\[
\begin{aligned}
SW(a)
&=
\sum_{i\in[k]}
\sum_{P\in\mathcal{P}_i}
\int_{a^{-1}(P)}
\bigl(v_j-\ell_P(f)\bigr)
\,d\mu_i(v)
\\
&=
\sum_{i\in[k]}
\mathcal{S}_i(f)
-
\sum_{P\in\mathcal{P}}
f_P\ell_P(f).
\end{aligned}
\]
The constraints on the aggregate flow are identical to those in
\cref{prop: max sw 0}. Therefore, optimizing over action profiles is
equivalent to \eqref{eq: upper 0}.
\end{proof}

\begin{remark}
An equivalent formulation can be obtained by introducing the edge
flows $f_e$ explicitly, as is standard in routing games
\cite{10.1145/509907.509971}.
\end{remark}

Since the latency term has already been separated in
\cref{prop: max sw}, the lower-level problem assigns agents only to
outcomes. Consequently, its Kantorovich dual potentials are
indexed by outcomes rather than strategies.

\begin{proposition}
\label{prop: kantorovich}
Fix an admissible flow $f$. A solution to \eqref{eq: lower} exists, and $\mathcal{S}_i(f)
=
\min_{p\in\mathbb{R}^{|D_i|}}
\Phi_i(p,f)$, where
\begin{equation}
\label{eq: dual phi}
\Phi_i(p,f)
:=
\int_{V_i}
\max_{j\in D_i}
(v_j-p_j)\,d\mu_i(v)
+
\sum_{j\in D_i}
p_j\nu_{f,i}(j).
\end{equation}

An optimal dual vector $\tau_i$ can be chosen such that $\tau_i \in [0, 2K]^{|D_i|}$ for some $K>0$ independent of $f$. For every optimal $\tau_i$
satisfying this bound, there exists an optimal transport map
$T:V_i\to D_i$ such that $T_{\#}\mu_i=\nu_{f,i}$ and
\[
T(v)
\in
\arg\max_{j\in D_i}
\{v_j-\tau_{i,j}\}
\]
for $\mu_i$-a.e. $v$, or equivalently, $T^{-1}(j)
=
V_{i,j}(\tau_i)$ up to a $\mu_i$-null set for every $j\in D_i$. Hence,
\begin{equation}
\label{eq: solution t}
\mathcal{S}_i(f)
=
\sum_{j\in D_i}
\int_{V_{i,j}(\tau_i)}
v_j\,d\mu_i(v).
\end{equation}
\end{proposition}

\begin{proof}
Apply \cref{lemma: transport map} (with $X=V_i,
\
Y=D_i,
\
\nu=\nu_{f,i},
\
c(v,j)=v_j$). Since $f$ is admissible,
\[
\sum_{j\in D_i}\nu_{f,i}(j)
=
r_i
=
\mu_i(V_i),
\]
so the source and target measures have equal mass. The lemma gives
\eqref{eq: dual phi} and, for every optimal dual vector $\tau_i$, an
optimal transport map satisfying $T(v)\in
\arg\max_{j\in D_i}(v_j-\tau_{i,j})$ for $\mu_i$-almost every $v$.

By \cref{footnote: disjoint}, the maximizer is unique for
$\mu_i$-almost every $v$. Hence, $T^{-1}(j)=V_{i,j}(\tau_i)$ a.e. for every $j\in D_i$, which gives \eqref{eq: solution t}.

The existence of an optimal dual vector in
$[0,2K]^{|D_i|}$ follows from the same normalization argument as in
the proof of \cref{prop: kantorovich 0}.
\end{proof}

\subsection{Comparative Statics and Convexity}

We next study the sensitivity of the lower-level value function $\mathcal{S}_i(f)$ and
the convexity of the social welfare optimization problem \eqref{eq: upper 0}.

For an admissible flow $f$, we reuse the notation $\Lambda_i(f)$ from \cref{subsubsec: existence} to denote the set of optimal dual potentials for this assignment problem:
\[
\Lambda_i(f)
:=
\arg\min_{p\in[0,2K]^{|D_i|}}
\Phi_i(p,f).
\]
When the objective in \eqref{eq: dual phi} is not strictly convex, $\Lambda_i(f)$ may contain multiple optimal solutions.

We let $\partial\mathcal{S}_i(f)$ denote the superdifferential
with respect to the group-flow vector
$f^i=(f_P)_{P\in\mathcal{P}_i}$.\footnote{
All differential statements refer to the concave extension
\[
\widetilde{\mathcal{S}}_i(x)
:=\min_{p\in[0,2K]^{|D_i|}}\Phi_i(p,x),
\qquad x\in\mathbb{R}^{\mathcal{P}_i},
\]
where $\Phi_i(p,x)$ is evaluated using
$\nu_{x,i}(j)=\sum_{P\in\mathcal{P}_{i,j}}x_P$.
This extension agrees with the original value function
on admissible flows. Superdifferentials are taken in
$\mathbb{R}^{\mathcal{P}_i}$, and we suppress the tilde.
}


\begin{lemma}
\label{lemma: envelope}
The value function $\mathcal{S}_i$ is concave in $f$. Moreover,
\begin{enumerate}
    \item
    \[
    \partial\mathcal{S}_i(f)
    =
    \left\{
    g\;\middle|\;
    \begin{array}{l}
    \exists\,\lambda_i\in\Lambda_i(f)
    \text{ such that}\\[1mm]
    g_P=\lambda_{i,j},
    \quad
    \forall j\in D_i,\ 
    \forall P\in\mathcal{P}_{i,j}
    \end{array}
    \right\}.
    \]

    \item If $\Lambda_i(f)$ is a singleton, then
    $\mathcal{S}_i$ is differentiable at $f$, with
    \[
    \frac{\partial\mathcal{S}_i(f)}{\partial f_P}
    =
    \lambda_{i,j},
    \qquad
    P\in\mathcal{P}_{i,j}.
    \]
\end{enumerate}
\end{lemma}

\begin{proof}
For each fixed $p$, $\Phi_i(p,f)$ is affine in $f$ because $\nu_{f,i}(j)
=
\sum_{P\in\mathcal{P}_{i,j}}f_P$. Hence, by \cref{prop: kantorovich}, $\mathcal{S}_i$ is the pointwise minimum of a family of affine
functions and is therefore concave.

Since $\Phi_i(p,f)$ is differentiable in $f$, Danskin's theorem
\cite[Proposition B.22]{bertsekas1997nonlinear} gives
\begin{equation}
\label{eq: conv}
\partial\mathcal{S}_i(f)
=
\operatorname{conv}
\left\{
\nabla_f\Phi_i(\lambda_i,f)
\;\middle|\;
\lambda_i\in\Lambda_i(f)
\right\}.
\end{equation}
For $P\in\mathcal{P}_{i,j}$,
\[
\frac{\partial\Phi_i(\lambda_i,f)}{\partial f_P}
=
\lambda_{i,j}.
\]
Moreover, $\Lambda_i(f)$ is convex, since it is the set of minimizers
of the convex function $\Phi_i(\cdot,f)$. The map
$\lambda_i\mapsto\nabla_f\Phi_i(\lambda_i,f)$ is linear, so its image
is convex. Therefore, the convex hull in \eqref{eq: conv} is
unnecessary, which proves the first claim.

If $\Lambda_i(f)$ is a singleton, then
$\partial\mathcal{S}_i(f)$ is a singleton, so $\mathcal{S}_i$ is
differentiable at $f$, giving the second claim.
\end{proof}

\begin{corollary}
\label{cor: convex}
The social welfare optimization problem \eqref{eq: upper 0} is a
convex optimization problem.
\end{corollary}

\begin{proof}
By \cref{lemma: envelope}, each $\mathcal{S}_i$ is concave. Moreover,
\[
\sum_{P\in\mathcal{P}}f_P\ell_P(f)
=
\sum_{e\in E}f_e\ell_e(f_e).
\]
Since each latency function is standard, $f_e\ell_e(f_e)$ is convex
in $f_e$. Hence, the objective in \eqref{eq: upper 0} is concave.
Since the feasible region is defined by linear constraints,
\eqref{eq: upper 0} is a convex optimization problem.
\end{proof}

\subsection{First-Order Characterization of Social Optima}

The convex structure established in \cref{cor: convex} yields a first-order characterization of social optima. Below, for capacity multipliers, we set $\delta_e=0$ whenever $c_e=\infty$ and adopt the convention $0\cdot\infty=0$ in the expressions $\delta_e c_e$ and $\delta_e(c_e-f_e)$.

\begin{theorem}
\label{prop: opt}
Given an instance $(E,\mathcal P,r,\ell,c,V,U,D)$ with separable utilities $U_i(x,y)=x-y$, a capacity-feasible flow $f$ is
an optimal solution to \eqref{eq: upper 0} if and only if there exist
$\tau$ and $\delta\geq0$ with $\tau_i\in\Lambda_i(f)$ for every $i\in[k]$, such that:
\begin{enumerate}
    \item for every $e\in E$,
    \begin{equation}
    \label{eq: cond 3}
    c_e>f_e
    \quad\Longrightarrow\quad
    \delta_e=0;
    \end{equation}

    \item for every $i\in[k]$,
    $P\in\mathcal{P}_{i,j}$, and
    $Q\in\mathcal{P}_{i,j'}$ with $f_P>0$,
    \begin{equation}
    \label{eq: cond 5}
    \tau_{i,j}
    -
    \delta_P
    -
    \hat{\ell}_P(f)
    \geq
    \tau_{i,j'}
    -
    \delta_Q
    -
    \hat{\ell}_Q(f),
    \end{equation}
    where $\delta_P:=\sum_{e\in P}\delta_e$.
\end{enumerate}
\end{theorem}

\begin{proof}
Introduce dual variables $\beta$, $\gamma$, and $\delta$ for the
nonnegativity, flow-conservation, and capacity constraints in
\eqref{eq: upper 0}, respectively. By \cref{cor: convex}, a
capacity-feasible flow $f$ is optimal if and only if the KKT
conditions hold. In addition to primal feasibility, these conditions
are
\begin{subequations}
\begin{align}
&\beta,\delta\geq0,&
\label{eq: dual feasibility}
\\
&\beta_Pf_P=0,
&\forall P\in\mathcal{P},
\label{eq: complementary 1}
\\
&\delta_e(c_e-f_e)=0,
&\forall e\in E,
\label{eq: complementary 2}
\\
&\bigl(\hat{\ell}_P(f)-\beta_P+\gamma_i+\delta_P\bigr)_{P\in\mathcal P_i}
\in\partial\mathcal S_i(f),
&\forall i\in[k].
\label{eq: stationarity}
\end{align}
\end{subequations}

By \cref{lemma: envelope}, every element of
$\partial\mathcal{S}_i(f)$ is constant over all strategies belonging to the same outcome. Hence, there exists
$\tau_i\in\Lambda_i(f)$ such that
\begin{equation}
\label{eq: stationarity tau}
\tau_{i,j}
-
\hat{\ell}_P(f)
+
\beta_P
-
\gamma_i
-
\delta_P
=
0,
\qquad
P\in\mathcal{P}_{i,j}.
\end{equation}

For the forward direction, suppose $f$ is optimal. Then
$\delta\geq0$, and \eqref{eq: cond 3} follows from complementary slackness \eqref{eq: complementary 2}. Fix
$P\in\mathcal{P}_{i,j}$ and
$Q\in\mathcal{P}_{i,j'}$ with $f_P>0$.
Then $\beta_P=0$, and \eqref{eq: stationarity tau} gives
\[
\tau_{i,j'}
-\delta_Q
-\hat{\ell}_Q(f)
=
\gamma_i-\beta_Q
\leq
\gamma_i
=
\tau_{i,j}
-\delta_P
-\hat{\ell}_P(f),
\]
which is \eqref{eq: cond 5}.

Conversely, suppose that $f$, $\tau$, and $\delta$ satisfy
\eqref{eq: cond 3}--\eqref{eq: cond 5}, where
$\tau_i\in\Lambda_i(f)$ for each $i$.
For each commodity $i$, choose
\[
P_i^*
\in
\arg\max_{P\in\mathcal{P}_i}
\left\{
\tau_{i,j(P)}
-
\hat{\ell}_P(f)
-
\delta_P
\right\},
\]
where $j(P)$ is the destination associated with $P$. Define
\begin{align}
\gamma_i
&=
\tau_{i,j(P_i^*)}
-
\hat{\ell}_{P_i^*}(f)
-
\delta_{P_i^*},
\\
\beta_P
&=
\hat{\ell}_P(f)
+
\delta_P
+
\gamma_i
-
\tau_{i,j(P)}.
\end{align}
By the definition of $P_i^*$, $\beta_P\geq0$ for every
$P\in\mathcal{P}_i$. If $f_P>0$, condition \eqref{eq: cond 5},
applied with $Q=P_i^*$, gives the reverse inequality and therefore
$\beta_P=0$. Thus, dual feasibility and complementary slackness hold.
Condition \eqref{eq: cond 3} gives complementary slackness for
$\delta$, and stationarity holds by construction. Hence, all KKT
conditions are satisfied, so $f$ is optimal.
\end{proof}

The variables $\delta_e$ are the Lagrange multipliers associated with
the capacity constraints and can be interpreted as capacity prices or
queueing delays, as in the side-constrained traffic-equilibrium
framework of Larsson and Patriksson \cite{LARSSON1999233}.
Condition \eqref{eq: cond 3} is the complementary slackness
condition: only binding capacity constraints can have positive prices
or queueing delays. Condition \eqref{eq: cond 5} states that, after
augmenting each path's marginal-cost latency by these terms, every
used path maximizes the corresponding threshold-adjusted utility.

When there are no capacity constraints ($c=\infty$), every multiplier
satisfies $\delta_e=0$. The characterization reduces to the classical
marginal-cost principle while allowing heterogeneous destination
valuations.

\begin{corollary}
\label{cor: opt no c}
Given an instance $(E,\mathcal P,r,\ell,c,V,U,D)$ with separable utilities $U_i(x,y)=x-y$ and $c=\infty$, an admissible flow
$f$ is socially optimal if and only if there exists
$\tau_i\in\Lambda_i(f)$ for every $i$ such that, for any
$P\in\mathcal{P}_{i,j}$ and
$Q\in\mathcal{P}_{i,j'}$ with $f_P>0$,
\[
\tau_{i,j}
-
\hat{\ell}_P(f)
\geq
\tau_{i,j'}
-
\hat{\ell}_Q(f).
\]
\end{corollary}

\subsection{Implementation of Social Optima}
Motivated by these optimality conditions, we next consider a congestion game with element-specific prices. For a fixed vector
$\delta\in\mathbb{R}_+^{|E|}$, define the effective element latency by
\[
\ell_e^\delta(f_e)
:=
\ell_e(f_e)+\delta_e.
\]
Thus, an agent of group $i$ and valuation $v$ choosing
$P\in\mathcal{P}_{i,j}$ receives utility
\begin{equation}
\label{eq: utility2}
u_{i,v}(P,f)
=
v_j-\ell_P(f)-\delta_P.
\end{equation}
We denote the resulting unconstrained priced congestion game by $(E,\mathcal P,r,\ell,\delta,V,U,D)$. The variables $\delta_e$ may be interpreted
either as monetary edge prices or as queueing delays.

The following result extends the threshold characterization of Nash
equilibria in \cref{prop: ne} to congestion games with resource prices.

\begin{corollary}
\label{cor: ne 2}
Given a priced instance $(E,\mathcal P,r,\ell,\delta,V,U,D)$ with separable utilities $U_i(x,y)=x-y$, a Nash equilibrium
is threshold-based. Moreover, an admissible flow $f$ is induced by a
Nash equilibrium if and only if the following hold for every group $i$:
\begin{enumerate}
    \item for any
    $P,Q\in\mathcal{P}_{i,j}$ with $f_P>0$,
    \[
    \ell_P(f)+\delta_P
    \leq
    \ell_Q(f)+\delta_Q;
    \]

    \item for any $j \in D_i$,
    \[
    \sum_{P\in\mathcal{P}_{i,j}}f_P
    =
    \mu_i(V_{i,j}(\tau_i)),
    \]
    where $\tau_{i,j}
    :=
    \min_{P\in\mathcal{P}_{i,j}}
    \{\ell_P(f)+\delta_P\}$.
\end{enumerate}
\end{corollary}

We next give an equivalent optimization characterization of a Nash
equilibrium. The objective combines the optimal assignment value with the negative of the Beckmann potential and
the linear resource-price term.

\begin{subequations}
\label{eq: upper with price}
\begin{alignat}{2}
\label{eq: upper obj 2}
\sup_f\quad
& \sum_{i\in[k]}\mathcal{S}_i(f)  - \sum_{e\in E} &&\left( \int_0^{f_e}\ell_e(x)\,dx + \delta_e f_e \right)\\
\label{eq: upper con1-2}
\text{s.t.}\quad
& f_P\geq0, && \quad \forall P\in\mathcal{P}, \\
\label{eq: upper con2-2}
& \sum_{P\in\mathcal{P}_i}f_P=r_i, && \quad \forall i\in[k].
\end{alignat}
\end{subequations}

The following result generalizes the classical optimization
characterization of Wardrop equilibria with queueing delays in
\cite[Theorem 3.1]{LARSSON1999233} to endogenous destination choices.

\begin{corollary}
\label{cor: NE delta}
Given a priced instance $(E,\mathcal P,r,\ell,\delta,V,U,D)$ with separable utilities $U_i(x,y)=x-y$, an admissible flow
$f$ is induced by a Nash equilibrium if and only if it maximizes
\eqref{eq: upper with price}.
\end{corollary}
\begin{proof}
Write the objective in \eqref{eq: upper with price} as
$F(f)=\sum_i\mathcal S_i(f)-C(f)$, where
\[
C(f):=\sum_{e\in E}\left(
\int_0^{f_e}\ell_e(x)\,dx+\delta_e f_e\right).
\]
Since each $\ell_e$ is continuous and nondecreasing, $C$ is
convex and differentiable, with
$\partial C(f)/\partial f_P=\ell_P(f)+\delta_P$.
By \cref{lemma: envelope}, $F$ is concave. The feasible set
is a product of simplices, so the KKT conditions are necessary
and sufficient for optimality.

Suppose first that $f$ maximizes $F$. By the KKT conditions
and \cref{lemma: envelope}, for each group $i$ there exists
$\lambda_i\in\Lambda_i(f)$ such that, whenever
$P\in\mathcal P_{i,j}$ and $f_P>0$,
\begin{equation}
\label{eq: priced kkt}
\lambda_{i,j}-\ell_P(f)-\delta_P
\geq
\lambda_{i,j'}-\ell_Q(f)-\delta_Q
\end{equation}
for every $j'\in D_i$ and $Q\in\mathcal P_{i,j'}$.

Choose an action profile $a$ that maximizes welfare among
profiles inducing $f$, and let $T_i(v)$ denote its outcome
assignment for group $i$. Fix the dual vector $\lambda_i$
obtained above. Since both $T_i$ and $\lambda_i$ are optimal
for the assignment problem, strong duality gives
\[
\begin{aligned}
0
&=\Phi_i(\lambda_i,f)-\mathcal S_i(f)\\
&=\int_{V_i}\Bigl[
\max_{j'\in D_i}\{v_{j'}-\lambda_{i,j'}\}
-\bigl(v_{T_i(v)}-\lambda_{i,T_i(v)}\bigr)
\Bigr]\,d\mu_i(v).
\end{aligned}
\]
The integrand is nonnegative and hence vanishes
$\mu_i$-almost everywhere. Thus, for almost every agent
assigned to outcome $j$,
\begin{equation}
\label{eq: lambda large}
v_j-\lambda_{i,j}
\geq v_{j'}-\lambda_{i,j'},
\qquad \forall j'\in D_i.
\end{equation}

Almost every agent uses a strategy $P$ with $f_P>0$.
Adding \eqref{eq: priced kkt} and \eqref{eq: lambda large} therefore gives
\[
v_j-\ell_P(f)-\delta_P
\geq v_{j'}-\ell_Q(f)-\delta_Q
\]
for every alternative strategy $Q\in\mathcal P_{i,j'}$.
Thus $a$ is a Nash equilibrium inducing $f$.

Conversely, suppose that $f$ is induced by a Nash equilibrium.
Define
\[
\tau_{i,j}:=\min_{P\in\mathcal P_{i,j}}
\{\ell_P(f)+\delta_P\}.
\]
Every used strategy attains this minimum for its outcome, and
almost every agent assigned to outcome $j$ maximizes
$v_j-\tau_{i,j}$. Consequently, the induced outcome assignment
attains the dual bound $\Phi_i(\tau_i,f)$ (see \eqref{eq: dual phi}), and hence is optimal: $\mathcal S_i(f)=\Phi_i(\tau_i,f)$.

For any admissible flow $h$, weak duality therefore yields
\[
\begin{aligned}
&\mathcal S_i(h)-\mathcal S_i(f)\\
&\quad\leq\sum_{j\in D_i}\tau_{i,j}
  \bigl(\nu_{h,i}(j)-\nu_{f,i}(j)\bigr)\\
&\quad\leq\sum_{P\in\mathcal P_i}
  \bigl(\ell_P(f)+\delta_P\bigr)(h_P-f_P).
\end{aligned}
\]
The last inequality uses
$\tau_{i,j}\leq\ell_P(f)+\delta_P$ for every
$P\in\mathcal P_{i,j}$, with equality whenever $f_P>0$.
By convexity of $C$,
\[
C(h)-C(f)\geq\sum_{P\in\mathcal P}
\bigl(\ell_P(f)+\delta_P\bigr)(h_P-f_P).
\]
Summing the preceding bounds over groups gives
$F(h)\leq F(f)$. Since $h$ was arbitrary, $f$ maximizes
\eqref{eq: upper with price}.
\end{proof}

We now combine the social-optimality characterization with the priced-game characterization. 

\begin{corollary}
\label{cor: induce so}
Given an instance $(E,\mathcal P,r,\ell,c,V,U,D)$ with separable utilities $U_i(x,y)=x-y$, let $f^*$ be a social optimum and
let $\delta^*$ be an associated vector of optimal Lagrange
multipliers for the capacity constraints in \eqref{eq: upper 0}.
Then $f^*$ is a Nash equilibrium of the priced congestion game $(E,\mathcal P,r,\hat{\ell},\delta^*,V,U,D)$.
\end{corollary}

\begin{proof}
By \cref{cor: NE delta}, a flow is a Nash equilibrium of
$(E,\mathcal P,r,\hat{\ell},\delta^*,V,U,D)$ if and only if it maximizes
\[
\sum_{i\in[k]}\mathcal{S}_i(f)
-
\sum_{e\in E}
\left(
\int_0^{f_e}\hat{\ell}_e(x)\,dx
+
\delta_e^* f_e
\right)
\]
over admissible flows. By the definition of the marginal-cost
function,
\[
\int_0^{f_e}\hat{\ell}_e(x)\,dx
=
f_e\ell_e(f_e).
\]
Therefore, up to the constant
$\sum_{e\in E}\delta_e^*c_e$, this objective is the Lagrangian of the
social welfare problem \eqref{eq: upper 0}:
\[
\sum_{i\in[k]}\mathcal{S}_i(f)
-
\sum_{e\in E}f_e\ell_e(f_e)
+
\sum_{e\in E}\delta_e^*(c_e-f_e).
\]
Since $f^*$ and $\delta^*$ form an optimal primal-dual pair, $f^*$
maximizes this Lagrangian over admissible flows. Hence, by
\cref{cor: NE delta}, $f^*$ is a Nash equilibrium of
$(E,\mathcal P,r,\hat{\ell},\delta^*,V,U,D)$.
\end{proof}

\section{Extensions}
\label{sec: Extensions}

\subsection{General Constraints}
\label{sec: convex}
We next extend our model to constraints on the aggregate element
flows $f_E:=(f_e)_{e\in E}$. Let $\mathcal C$ be a finite collection
of continuous functions, with constraints $g(f_E)\leq0$ for
$g\in\mathcal C$. An admissible flow is \emph{constraint-feasible}
if it satisfies all these inequalities. We assume that the set
of constraint-feasible admissible flows is nonempty. The
edge-capacity model in \cref{sec: model} is recovered by taking
$g_e(f_E)=f_e-c_e$ for each finite-capacity element $e \in E$.

A constraint $g$ \emph{involves} a strategy $P$ if it depends on
$f_e$ for some $e\in P$. A strategy $P$ is \emph{feasible at $f$}
if $g(f_E)\leq0$ for every constraint involving $P$. Given an
action profile $a$ with induced flow $f$, define
\[
\begin{aligned}
&u_{i,v}(a_{i,v},a_{-(i,v)})\\
&=\begin{cases}
U_i(v_j,\ell_P(f)),
&a_{i,v}=P\in\mathcal P_{i,j},
            P\text{ is feasible at }f,\\
-\infty, &\text{otherwise}.
\end{cases}
\end{aligned}
\]
We retain the same Nash best-response condition and restrict
attention to equilibria with constraint-feasible induced flows.
Social optima maximize welfare over constraint-feasible action
profiles.\footnote{An upper-capacity violation affects agents currently
using the resource, who can profitably deviate if they have an
unconstrained outside option with finite utility. A lower-bound
violation, however, can occur when nobody uses the resource.
For example, consider a one-resource game with unit mass,
an unconstrained outside option with finite utility, and the
constraint $b-f_e\leq0$, where $0<b<1$. The all-outside profile
is an infeasible Nash equilibrium: a unilateral entrant leaves
$f_e=0$ and receives utility $-\infty$.}

At a constraint-feasible flow, every strategy is feasible, and
a unilateral deviation leaves the aggregate flow unchanged.
Consequently, the Nash characterization in
\cref{sec: Characterizations} applies to constraint-feasible
equilibria with the same outcome thresholds, $\tau_{i,j}(f):=\min_{P\in\mathcal P_{i,j}}\ell_P(f)$.

For each constraint-feasible flow, the lower-level assignment
problem and its dual characterization are unchanged. At the
upper level, replace the element-capacity constraints by
$g(f_E)\leq0$ for all $g\in\mathcal C$.

Under the separable-utility model in \cref{sec: separable},
suppose that each $\ell_e$ is continuously differentiable,
$x\mapsto x\ell_e(x)$ is convex, and every $g\in\mathcal C$
is convex and continuously differentiable. The upper-level
problem is then a convex optimization problem. Under a suitable constraint qualification, optimal multipliers yield resource prices that support a socially optimal flow. Details appear in \cref{app: General Constraints}.

\subsection{General Measures}
\label{subsec: abs conti}

We now remove the assumption in \cref{subsec: val and act} that each
measure $\mu_i$ is absolutely continuous with respect to Lebesgue
measure. Instead, we allow $\mu_i$ to be an arbitrary finite measure.
In particular, this permits a positive mass of agents to share exactly
the same type.

To accommodate atoms while retaining the nonatomic interpretation of
the underlying population, we allow the mass associated with each type
to be split across strategies. Accordingly, we define an action
profile to be a measurable function $a:V\to\Delta(\mathcal{P})$ such that, for every $(i,v)\in V$, the distribution $a((i,v))$ places
all its mass on $\mathcal{P}_i$. Equivalently, for each group $i$, an
action profile specifies a measurable stochastic kernel $a_i(P \mid v)$ for $P\in\mathcal{P}_i$, satisfying $a_i(P\mid v)\geq 0$ and $\sum_{P\in\mathcal{P}_i}a_i(P\mid v)=1$ for $\mu_i$-almost every $v$.

The flow induced by an action profile is
\[
f_P
=
\int_{V_i}a_i(P\mid v)\,d\mu_i(v),
\qquad
P\in\mathcal{P}_i.
\]
The following lemma extends \cref{lemma: admissible flow}.

\begin{lemma}
\label{lemma: general measure}
Every randomized action profile induces an admissible flow.
Conversely, every admissible flow is induced by some randomized action
profile.
\end{lemma}

\begin{proof}
Every randomized action profile induces nonnegative flows
satisfying
\[
\sum_{P\in\mathcal P_i}f_P
=
\int_{V_i}\sum_{P\in\mathcal P_i}a_i(P\mid v)\,d\mu_i(v)
=
r_i.
\]
Conversely, given an admissible flow $f$, set
$a_i(P\mid v):=f_P/r_i$. This defines a probability
distribution on $\mathcal P_i$ and induces $f$, since
$\mu_i(V_i)=r_i$.
\end{proof}

The characterizations of Nash equilibria and social optima require minor modifications. See \cref{app: General Measures} for more details.

This extension is distinct from an atomic congestion game \cite{ROUGHGARDEN2004389}. An atom of $\mu_i$ represents a positive mass of infinitesimal agents sharing the same type, rather than a single player of positive weight. Thus, an atom represents a mass of identical agents whose flow may be split across multiple strategies.

\section{Conclusion}
\label{sec: Conclusion}
Motivated by applications in advanced air mobility (AAM), we introduced a class of nonatomic congestion games with heterogeneous outcome valuations. Despite the infinite-dimensional population of agents, we derived finite-dimensional characterizations of both Nash equilibria and social optima. Nash equilibria are characterized by aggregate
flows and outcome-level threshold vectors, while social optima are characterized by aggregate flows and strategy-level Kantorovich dual potentials. These characterizations reveal a geometric structure in the assignment of heterogeneous agents that is closely connected to semi-discrete optimal transport. In the separable-utility setting, we further showed how social optima can be induced through marginal-cost latencies and capacity prices or queueing delays.

In the AAM application, the equilibrium characterization in
\cref{prop: ne} describes how heterogeneous travelers sort
across vertiports and routes according to their valuations
and the resulting congestion. Aggregate flows and outcome
thresholds describe this sorting without explicitly tracking
individual agents, while the Kantorovich dual potentials in
\cref{prop: kantorovich 0} characterize welfare-optimal assignments
for given aggregate flows. In the separable-utility setting,
\cref{cor: induce so} further shows how marginal-cost tolls
on congestible resources, together with landing fees at
capacity-constrained vertiports, can support a socially
optimal allocation. These landing fees are given by the
multipliers associated with the vertiport capacity constraints.

This work focuses on a new congestion game model and its finite-dimensional characterizations. Future work includes studying the price of anarchy and developing computational methods for Nash equilibria and social optima. On the application side, an important direction is to design dynamic
mechanisms that steer equilibrium behavior toward socially optimal outcomes without requiring explicit knowledge of the underlying valuation distributions. Extending the framework to other resource-allocation environments is another promising direction.

\section*{Acknowledgment}
The authors used GPT-6 Astra to assist with checking and revising the proofs in \cref{sec: Characterizations,sec: separable,sec: Extensions} and \cref{app: proof}.
\section*{References}
\bibliographystyle{IEEEtran}
\bibliography{notes/refs}

\appendices

\section{Proof of \cref{lemma: transport map}}
\label[appendix]{app: proof}
\begin{proof}
The Kantorovich problem
\[
\sup_{\pi\in\Pi(\mu,\nu)}
\int_{X\times Y}c(x,j)\,d\pi(x,j)
\]
has the dual formulation
\[
\inf_{\phi,q}
\left\{
\int_X\phi(x)\,d\mu(x)
+
\sum_{j=1}^m q_j\nu_j
\right\}
\]
subject to
\[
\phi(x)+q_j\geq c(x,j),
\qquad
\forall x\in X,\ j\in Y.
\]
By Kantorovich duality
\cite[Theorem 5.10]{villani2008optimal}, the primal and dual values
coincide and both optima are attained.

For a fixed $q$, the smallest feasible choice of $\phi$ is
\[
\phi_q(x)
=
\max_{j\in Y}\bigl(c(x,j)-q_j\bigr).
\]
Substituting $\phi_q$ into the dual objective gives
\eqref{eq: semi discrete duality}.

Fix an optimal dual solution $q$ and an optimal transport plan $\pi$.
Define
\[
M(x)
:=
\arg\max_{j\in Y}
\bigl(c(x,j)-q_j\bigr).
\]
Since $Y$ is finite and $c(\cdot,j)$ is measurable for every $j$,
$M$ is a measurable correspondence with nonempty finite values.

By primal-dual optimality,
\[
\int_{X\times Y}c(x,j)\,d\pi(x,j)
=
\int_X\phi_q(x)\,d\mu(x)
+
\sum_{j=1}^m q_j\nu_j.
\]
Using the marginal constraints of $\pi$, we obtain
\[
\int_{X\times Y}
\left[
\phi_q(x)+q_j-c(x,j)
\right]
d\pi(x,j)
=
0.
\]
The integrand is nonnegative. Hence, $\phi_q(x)+q_j=c(x,j)$
for $\pi$-a.e. $(x,j)$. Equivalently, $j\in M(x)$ for $\pi$-a.e. $(x,j)$. Thus, $\pi$ is concentrated on the graph of
$M$.

Disintegrating $\pi$ with respect to $\mu$
(cf. \cite[Theorem 4.1.17]{Durrett_2019}), there exist measurable
functions
\[
\lambda_j:X\to[0,1],
\qquad
j\in Y,
\]
such that $\sum_{j=1}^m\lambda_j(x)=1$ for $\mu$-a.e. $x$ and $\pi(dx,j)=\lambda_j(x)\,\mu(dx)$.
Since $\pi$ is concentrated on the graph of $M$, $\lambda_j(x)=0$ for $\mu$-a.e. $x$ such that $j\notin M(x)$.
Moreover, since the second marginal of $\pi$ is $\nu$,
\[
\int_X\lambda_j(x)\,d\mu(x)=\nu_j,
\qquad
j\in Y.
\]

Since $\mu$ is nonatomic and $Y$ is finite, a consequence
of Lyapunov's convexity theorem
(cf. \cite[Theorem 13.33]{aliprantis2006infinite})
yields a measurable selection $T$ of $M$, up to a null set,
with
\[
\mu(T^{-1}(j))
=\int_X\lambda_j(x)\,d\mu(x)
=\nu_j,
\qquad j\in Y.
\]
Therefore, $T(x)\in M(x)$ for $\mu$-a.e. $x$, and $T$ has the prescribed target distribution.

Finally, since $T(x)\in M(x)$ for $\mu$-a.e. $x$, $c(x,T(x))
=
\phi_q(x)+q_{T(x)}$ almost everywhere. Therefore,
\begin{align*}
\int_X c(x,T(x))\,d\mu(x)
&=
\int_X\phi_q(x)\,d\mu(x)
+
\sum_{j=1}^m q_j\mu(T^{-1}(j))\\
&=
\int_X\phi_q(x)\,d\mu(x)
+
\sum_{j=1}^m q_j\nu_j.
\end{align*}
The right-hand side is the optimal dual value. By
\eqref{eq: semi discrete duality}, it is also the optimal primal
value. Hence, the transport plan induced by $T$ is optimal.
\end{proof}

\section{Some Single-Destination Examples}
\label[appendix]{app: some examples}
In \cref{subsec: ex 1,subsec: ex 2,subsec: ex 3}, we use the one-link network shown in \cref{fig:one-link} to illustrate the role of heterogeneous valuations and capacity constraints in a simple setting. These examples are intended to build intuition for the model and the results developed in the main text. A summary is provided in \cref{table: compare}.

A similar one-link example was used in \cite[Section 1]{10.5555/1109557.1109630} to demonstrate the inefficiency of Nash equilibria in routing games. Our examples instead focus on illustrating the impact of heterogeneous valuations and capacity constraints.

Throughout this section, we focus on a single commodity, so we drop the index $i$ to simplify notation.
\subsection{First Example: Classical Routing Games}
\label[appendix]{subsec: ex 1}
We begin with a benchmark example that essentially reduces to a classical nonatomic routing game. Specifically, every agent prefers traveling to dropping out, and the edge capacity is sufficiently large so that it never binds.
\begin{example}
\label{ex: 1}
Consider a one-link network with total traffic $r=1$, latency function $\ell_1(x) = x$, and edge capacity $c_1 = 1$, as shown in \cref{fig:one-link}. Agent valuations are uniformly distributed on $[2,3]$, i.e., $\mu \sim \text{Unif}([2, 3])$. The unique Nash equilibrium routes all traffic through the link, yielding a social welfare of 3/2. This equilibrium is also socially optimal.
\end{example}

\begin{figure}
    \centering
    \begin{tikzpicture}[->, thick, node distance=2.5cm]
        \node[circle, draw, minimum size=0.7cm, inner sep=1pt] (A) {$s$};
        \node[circle, draw, minimum size=0.7cm, inner sep=1pt, right=of A] (B) {$t_1$};
        \node[left=2mm of A] {$r=1$};
        \draw (A) -- node[midway, above] {$\begin{aligned}
                &c_1 = 1 \text{ or } 1/2\\
                &\ell_1(x) = x
            \end{aligned}$} (B);
    \end{tikzpicture}
    \caption{A one-link network with total traffic $r = 1$, latency function $\ell_1(x) = x$, and capacity $c_1$. This network is used throughout \cref{ex: 1,ex: 2,ex: no nash}.}
    \label{fig:one-link}
\end{figure}
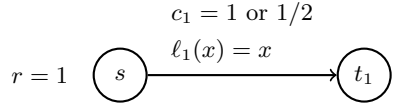

\begin{table}
\caption{Pricing interventions that induce a socially optimal flow
in \cref{ex: 1,ex: 2,ex: no nash}.}
\label{table: compare}
\centering
\small
\setlength{\tabcolsep}{3pt}
\begin{tabular}{|c|c|c|c|}
\hline
Example & \shortstack{Valuation\\distribution} & Capacity
& \shortstack{Pricing\\intervention} \\
\hline
\cref{ex: 1} & $\operatorname{Unif}([2,3])$ & $c_1=1$
& none \\
\hline
\cref{ex: 2} & $\operatorname{Unif}([0,1])$ & $c_1=1$
& \shortstack{marginal-cost toll\\$t_1(x)=x$} \\
\hline
\cref{ex: no nash} & $\operatorname{Unif}([2,3])$ & $c_1=1/2$
& \shortstack{marginal-cost toll\\and capacity fee\\
  $t_1(x)=x$, $\delta^*=3/2$} \\
\hline
\end{tabular}
\end{table}

\subsection{Second Example: Heterogeneous Valuations and Marginal-Cost Tolling}
\label[appendix]{subsec: ex 2}
We now consider a setting in which agents' valuations are sufficiently heterogeneous that some agents may prefer to drop out. As in classical routing games with elastic (or variable) demand, a Nash equilibrium need not maximize social welfare because agents do not internalize the congestion externality. The standard remedy is marginal-cost tolling, which we illustrate below.
\begin{example}
\label{ex: 2}
Consider routing one unit of traffic through the network in \cref{fig:one-link} again. Unlike \cref{ex: 1}, we now let the valuation distribution be $\mu \sim \text{Unif}([0, 1])$. For convenience, we refer to agents with valuations at least 1/2 as \emph{high-valuation}  agents and the remainder as \emph{low-valuation} agents. 

The unique Nash equilibrium routes exactly the high-valuation agents while the low-valuation agents drop out. Indeed, if the traffic exceeds 1/2, then some low-valuation agents obtain negative utility and prefer to drop out. Conversely, if the traffic is below 1/2, then some high-valuation agents obtain positive utility and have an incentive to enter. Hence, the unique equilibrium flow is 1/2.

Now, consider the social optimum. As shown in \cref{prop: kantorovich}, any social optimum has a threshold structure: there exists $\tau$ such that agents with valuations below $\tau$ drop out while those with valuations above $\tau$ travel. The corresponding social welfare is $\int_\tau^1 udu - (1-\tau)^2$, where the first term is the total valuation of the participating agents and the second is the total congestion cost. Differentiating with respect to $\tau$ gives $\tau = 2/3$.

As in classical routing games, the socially optimal flow can be induced through marginal-cost tolling. Replacing the latency function $\ell_1(x) = x$ by its marginal-cost counterpart $\hat{\ell}_1(x) = 2x$ yields a unique Nash equilibrium in which exactly the highest one third of the agents choose to travel, coinciding with the social optimum under the original latency function.
\end{example}

\subsection{Third Example: Capacity Constraints and Capacity Fees}
\label[appendix]{subsec: ex 3}
We next consider a setting in which the capacity constraint becomes binding. Unlike the previous example, the capacity constraint may prevent the existence of a Nash equilibrium. A natural way to restore equilibrium is to introduce prices associated with the capacity constraints, analogous to competitive prices in assignment games \cite{shapley1971assignment}. We refer to these prices as \emph{capacity fees}.
\begin{example}
\label{ex: no nash}
Consider the setting in \cref{ex: 1} ($\mu \sim \text{Unif}([2, 3])$) but decrease the capacity to $1/2$. Any flow that routes half of the traffic through the link and the remaining mass \(1/2\) through the outside option is capacity-feasible. \emph{However, there is no Nash equilibrium.} 
Suppose the induced flow is capacity-feasible. If its total traffic is strictly below the capacity, then some non-routed agent can profitably deviate by entering. If the total traffic equals the capacity, then any non-routed agent still has a profitable deviation, since her unilateral deviation has measure zero and therefore does not violate the capacity constraint. Conversely, if the induced flow exceeds the capacity, every routed agent receives utility $-\infty$ and prefers to drop out. Hence, no Nash equilibrium exists.

Allowing mixed strategies does not directly restore the standard existence argument, since our game violates \cite[Condition (a)]{schmeidler1973equilibrium}. In contrast, under the equilibrium notion of \cite[Definition 4.3.1, Proposition 4.3.2]{10.5555/1076293}, where only coalitions of positive measure are allowed to deviate, every capacity-feasible flow routing half of the traffic is an equilibrium. Thus, the nonexistence result stems from the equilibrium concept rather than the capacity constraint alone. We discuss this distinction further in \cref{app: different equilibrium notions}.

A socially optimal allocation routes precisely the agents with valuations
above $5/2$, yielding a routed mass of $x^*=1/2$. Indeed, for a given
routed mass $x$, welfare is maximized by serving the agents with the
highest valuations, giving
\[
W(x)=\int_{3-x}^{3}v\,dv-x\ell_1(x)
=3x-\frac{3}{2}x^2.
\]
Since $W$ is strictly increasing on $[0,1/2]$, the capacity constraint
binds at the optimum. Its optimal multiplier is
\[
\delta^*=W'(x^*)=3-3x^*=\frac{3}{2}.
\]

As in \cref{cor: induce so}, this optimum can be implemented by
combining the marginal-cost toll $t_1(x)=x\ell_1'(x)=x$ with the
capacity fee $\delta^*=3/2$. A traveling agent then experiences the
total cost
\[
\ell_1(x)+t_1(x)+\delta^*
=\hat{\ell}_1(x)+\delta^*
=2x+\frac{3}{2}.
\]
The resulting unique Nash flow is $x^*=1/2$, at which the total cost
is $5/2$ and exactly the high-valuation half of the agents travels,
up to the measure-zero set of indifferent agents. At this flow, the
total monetary charge is $t_1(x^*)+\delta^*=2$: the marginal-cost
toll is $1/2$, and the capacity fee is $3/2$.
\end{example}

\subsection{Relationship with Existing Single-Destination Traffic Assignment Models}
\label[appendix]{app: NE with single des}
Our model, together with the extension to general convex constraints in \cref{sec: convex}, subsumes classical traffic assignment models with elastic demand and side constraints \cite{YILDIRIM2005659}. Below, we compare our equilibrium notion with those used in the single-destination literature.

\subsubsection{Routing Games with Elastic Demand}
When capacity constraints are absent ($c_e=\infty$ for all $e\in E$), our equilibrium notion (\cref{def: NE}) is equivalent to the equilibrium definition of \cite[Definition 2.1]{10.5555/1109557.1109630}. Conversely, \cref{def: NE} together with \cref{prop: ne} provides a micro-foundation for their aggregate-flow formulation by explicitly modeling a continuum of heterogeneous agents.

\subsubsection{Routing Games with Edge Capacity Constraints}
\label[appendix]{app: different equilibrium notions}
We now consider routing games with edge capacity constraints and inelastic demand. Since the destination choice becomes trivial in this setting, we denote an instance simply by $(G,r,\ell,c)$. In the remainder of this subsection, we identify a Nash equilibrium with its induced flow. That is, a capacity-feasible flow is called a \emph{Nash equilibrium} if it is induced by an action profile satisfying \cref{def: NE}. The literature contains two other equilibrium concepts for this setting. Below, we follow the terminology in Correa et al. \cite{doi:10.1287/moor.1040.0098}. 

We first define the \emph{capacitated user equilibrium} introduced in \cite[Definition 3.1]{doi:10.1287/moor.1040.0098}. The following formulation is introduced by \cite[Definition 4.3.1]{10.5555/1076293}.
\begin{definition}
\label{def: NE TR}
Given an instance $(G, r, \ell, c)$, a capacity-feasible flow $f$ is a \emph{capacitated user equilibrium} if for all $i \in [k]$, $P_1, P_2 \in \mathcal{P}_i$ with $f_{P_1} > 0$, and $\delta \in (0, f_{P_1}]$, either $\ell_{P_1}(f) \leq \ell_{P_2}(\tilde{f})$, where 
$$
\tilde{f}_P = \begin{cases}
f_P - \delta, &\text{ if } P = P_1\\
f_P + \delta, &\text{ if } P = P_2\\
f_P, &\text{ if } P \notin \{P_1, P_2\},
\end{cases}
$$
or $\tilde{f}$ is not capacity-feasible. 
\end{definition}

Unlike \cref{def: NE}, a capacitated user equilibrium only rules out deviations by measurable sets of agents of positive measure. \cref{ex: no nash} shows how this difference in definitions leads to different existence conclusions. The following result from \cite[Proposition 4.3.2]{10.5555/1076293} gives an equivalent characterization of a capacitated user equilibrium.
\begin{proposition}
Given an instance $(G, r, \ell, c)$, a capacity-feasible flow $f$ is a capacitated user equilibrium if and only if for every $i \in [k]$ and $P_1, P_2 \in \mathcal{P}_i$ with $f_{P_1} > 0$, one of the following holds.
\begin{enumerate}
    \item $\ell_{P_1}(f) \leq \ell_{P_2}(f)$, or
    \item there is an edge $e \in P_2 \setminus P_1$ with $f_e = c_e$.
\end{enumerate}
\end{proposition}

The second notion is the \emph{BMW equilibrium} defined in \cite[Section 3.2]{doi:10.1287/moor.1040.0098}; see also Larsson and Patriksson \cite{LARSSON1999233,LARSSON1995433}.\footnote{Contradicting the terminology in Correa et al. \cite{doi:10.1287/moor.1040.0098}, the BMW equilibrium is called a \emph{capacitated user equilibrium} in Larsson and Patriksson \cite{LARSSON1995433}.} It is defined as a solution to \eqref{eq: upper new}. This equilibrium notion admits an optimization-based characterization and, whenever the capacity-feasible set is nonempty, an equilibrium exists by the extreme value theorem.
\begin{subequations}
\label{eq: upper new}
\begin{align}
\min_{f} \quad  &\sum_{e \in E}\int_0^{f_e} \ell_e(x)dx \\
\ \ \text{s.t. } \quad &f_P \geq 0, \forall P \in \mathcal{P}\\
&\sum_{P \in \mathcal{P}_i} f_P = r_i, \forall i \in [k]\\
&f_e \leq c_e, \forall e \in E\\
&f_e = \sum_{P \in \mathcal{P}: e \in P} f_P, \forall e \in E.
\end{align}
\end{subequations}

\begin{definition}
\label{def: NE LP}
Given an instance $(G, r, \ell, c)$, a capacity-feasible flow $f$ is a BMW equilibrium if it minimizes \eqref{eq: upper new}.
\end{definition}

Correa et al. \cite[Lemma 3.2]{doi:10.1287/moor.1040.0098} (see also \cite[Lemma 4.3.6]{10.5555/1076293}) proved that every BMW equilibrium is a capacitated user equilibrium. The converse is false by \cite[Example 4.3.3]{10.5555/1076293}.
Below, we show that our definition is a further restriction of a BMW equilibrium. \cref{fig: Venn for Nash} summarizes the relationship.
\begin{lemma}
Given an instance $(G,r,\ell,c)$, every Nash equilibrium is a BMW equilibrium. The converse is false.
\end{lemma}
\begin{proof}
Suppose that $f$ is a Nash equilibrium induced by an action profile $a$. We show that $f$ satisfies the variational inequality characterizing a BMW equilibrium.

Let $g$ be any capacity-feasible flow. Since $g$ is admissible,
\cref{lemma: admissible flow} implies that there exists an action profile $b$ inducing $g$. Since $a$ is a Nash equilibrium, for $\mu$-almost every agent $(i,v)$,
$$
u_{i,v}(a_{i,v},a_{-(i,v)}) \ge u_{i,v}(b_{i,v},a_{-(i,v)}).
$$
Because there is only one destination, utilities are given by
$u_{i, v}(P, a_{-(i, v)}) = -\ell_P(f)$. Thus, $\ell_{a_{i,v}}(f)\le \ell_{b_{i,v}}(f)$ for $\mu$-almost every $(i, v)$.

Integrating $\ell_{a_{i,v}}(f)\le \ell_{b_{i,v}}(f)$  over the valuation space gives $\int_V\ell_{a_{i,v}}(f) d\mu \leq \int_V\ell_{b_{i,v}}(f) d\mu$. Since $a$ induces $f$ and $b$ induces $g$, we get
$$
\sum_{P\in\mathcal P}\ell_P(f)f_P
    \le
\sum_{P\in\mathcal P}\ell_P(f)g_P,
$$
or equivalently,
$$
\sum_{P\in\mathcal P}
\ell_P(f)(g_P-f_P)\ge0.
$$
This is precisely the variational inequality characterizing the minimizers of \eqref{eq: upper new}. Therefore, $f$ is a BMW equilibrium.

The converse does not hold, as can be seen by modifying \cref{ex: no nash}.\footnote{The example can be modified to the inelastic-demand setting: replace the outside option by an alternative edge with constant cost $1$. The BMW equilibrium routes $1/2$ unit on each edge, but it is not a Nash equilibrium under \cref{def: NE}, since an agent using the cost-$1$ edge can unilaterally switch to the capacity-constrained edge, whose latency is $1/2$.}
\end{proof}

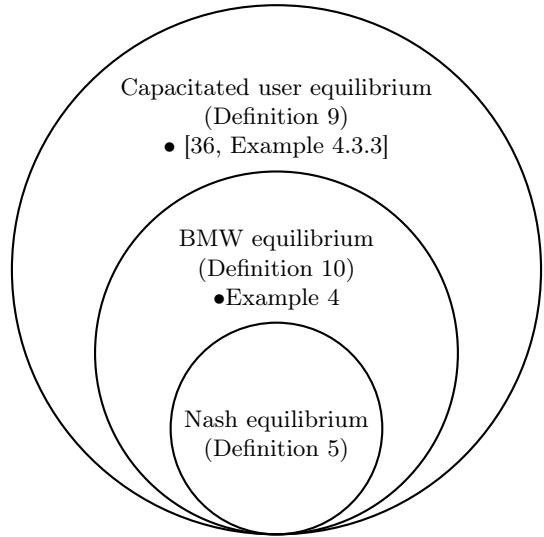
\begin{figure}[h]
\centering
\begin{tikzpicture}
    \draw[thick] (0,0) circle (3.5cm) node[above right=2mm] {};

    \draw[thick] (0,-1.1) circle (2.4cm) node[left=3mm] {};

    \draw[thick] (0,-2.1) circle (1.4cm) node[right=3mm] {};

    \node[align=center] at (0,2) {Capacitated user equilibrium \\
    (\cref{def: NE TR}) \\ $\bullet$ \cite[Example 4.3.3]{10.5555/1076293}};

    \node[align=center] at (0,-2.2) {Nash equilibrium \\
    (\cref{def: NE})};
    
    \node[align=center] at (0,0) {BMW equilibrium \\
    (\cref{def: NE LP}) \\ $\bullet$ (Modified) \cref{ex: no nash}};

\end{tikzpicture}
\caption{Relationship among the three equilibrium notions for routing games with edge-capacity constraints. The cited examples show that each inclusion is strict.}
\label{fig: Venn for Nash}
\end{figure}

\section{Basics of Lattice Theory}
\label[appendix]{app: lattice}

We review the basic definitions and results from lattice theory used
in the paper, following \cite[Chapter 2]{topkis1998supermodularity}.

A \emph{partially ordered set} is a set $X$ equipped with a binary
relation $\succeq$ that is reflexive, antisymmetric, and transitive.
For $x^1,x^2\in X$, their least upper bound, when it exists, is called
their \emph{join} and is denoted by $x^1\vee x^2$. Their greatest
lower bound is called their \emph{meet} and is denoted by
$x^1\wedge x^2$.

A partially ordered set $X$ is a \emph{lattice} if
$x^1\vee x^2$ and $x^1\wedge x^2$ exist in $X$ for every
$x^1,x^2\in X$. A lattice is \emph{complete} if every subset of $X$
has a supremum and an infimum. A subset $X'\subseteq X$ is a
\emph{sublattice} of $X$ if
\[
x^1\vee x^2,\;x^1\wedge x^2\in X'
\qquad
\forall x^1,x^2\in X',
\]
where the join and meet are taken with respect to $X$.

For $\mathbb{R}^n$ endowed with the product order,
$x\succeq y$ if and only if $x_i\geq y_i$ for every $i$. The join and meet are then given by the componentwise maximum and minimum, respectively.

A function $f:X\to\mathbb{R}$ on a lattice $X$ is
\emph{supermodular} if
\[
f(x^1\vee x^2)+f(x^1\wedge x^2)
\geq
f(x^1)+f(x^2)
\]
for all $x^1,x^2\in X$. It is \emph{submodular} if the inequality is
reversed. A standard result
\cite[Theorem 2.7.1]{topkis1998supermodularity} states that the set of
maximizers of a supermodular function is a sublattice whenever it is
nonempty. Equivalently, the set of minimizers of a submodular function
is a sublattice.

Finally, every nonempty compact sublattice of $\mathbb{R}^n$ is a complete lattice
\cite[Theorem 2.3.1]{topkis1998supermodularity}.

\section{Details of the Extensions}

\subsection{General Constraints}
\label[appendix]{app: General Constraints}
Under the assumptions in \cref{sec: convex}, the upper-level
problem is a convex optimization problem. For the pricing
result, assume a constraint qualification ensuring the existence
of KKT multipliers, such as Slater's condition relative to the
admissible-flow set. Let $f^*$ be a social optimum, and let
$\eta_g^*\geq0$ be associated multipliers satisfying
$\eta_g^*g(f_E^*)=0$. Define the fixed resource prices
\[
q_e^*:=\sum_{g\in\mathcal C}\eta_g^*
\frac{\partial g}{\partial f_e}(f_E^*),
\qquad e\in E.
\]
The constraint contribution to the stationarity condition for
strategy $P$ is $q_P^*:=\sum_{e\in P}q_e^*$. Thus, the argument
of \cref{cor: induce so} shows that $f^*$ can be induced by a
Nash equilibrium of the game without hard constraints and with
perceived resource costs
\[
\hat\ell_e(x)+q_e^*
=\ell_e(x)+x\ell_e'(x)+q_e^*.
\]
For element-capacity constraints, $q_e^*$ reduces to the
corresponding capacity fee. For general constraints, these
prices may be negative and are then interpreted as subsidies;
they are nonnegative when every constraint function is
coordinatewise nondecreasing. This is the welfare-implementation
analogue of the multiplier-based pricing characterization in
\cite[Theorem 2.2]{LARSSON1999233}.

\subsection{General Measures}
\label[appendix]{app: General Measures}
For a threshold
vector $\tau_i$, define the best-response correspondence
\[
B_i(v;\tau_i)
:=
\arg\max_{j\in D_i}
U_i(v_j,\tau_{i,j}).
\]
Unlike in the absolutely continuous case, $B_i(v;\tau_i)$ need not be
a singleton for $\mu_i$-almost every $v$. Consequently, the threshold
regions
\[
V_{i,j}(\tau_i)
=
\{v\in V_i:j\in B_i(v;\tau_i)\}
\]
need not be a.e. disjoint.

Accordingly, the mass condition in the Nash equilibrium
characterization is replaced by the existence of a measurable
destination kernel $\kappa_i(j\mid v)$ satisfying
$\kappa_i(j\mid v)\geq0$,
$\sum_{j\in D_i}\kappa_i(j\mid v)=1$, and
\[
\kappa_i(j\mid v)>0
\quad\Longrightarrow\quad
j\in B_i(v;\tau_i)
\]
for $\mu_i$-almost every $v$, together with
\[
\sum_{P\in\mathcal{P}_{i,j}}f_P
=
\int_{V_i}\kappa_i(j\mid v)\,d\mu_i(v),
\qquad
\forall j\in D_i.
\]
Thus, the mass associated with a type that is indifferent among
several outcomes may be split among those outcomes. When $\mu_i$ is
absolutely continuous, the indifference sets have measure zero by
\cref{footnote: disjoint}, and the condition reduces to
\[
\sum_{P\in\mathcal{P}_{i,j}}f_P
=
\mu_i(V_{i,j}(\tau_i)),
\qquad
\forall j\in D_i,
\]
recovering the equilibrium characterization in
\cref{sec: Characterizations}.

A similar modification applies to the social-welfare problem.
Kantorovich duality continues to provide a finite-dimensional dual
characterization of the lower-level assignment problem. However, an optimal coupling
need not, in general, be induced by a deterministic transport map.
Instead, an optimal coupling can be represented by a stochastic
kernel. In particular, if $\pi_i$ is an optimal coupling, applying the disintegration theorem (cf. \cite[Theorem 4.1.17]{Durrett_2019}) to $\pi_i$ with respect to $\mu_i$ yields a measurable kernel $\kappa_i(P\mid v)$ such that
\[
\pi_i(dv,\{P\})
=
\mu_i(dv)\kappa_i(P\mid v).
\]
The kernel assigns positive mass only to strategies supported by the
optimal dual potentials, $\mu_i$-almost everywhere, and satisfies
\[
\int_{V_i}\kappa_i(P\mid v)\,d\mu_i(v)
=
f_P,
\qquad
P\in\mathcal{P}_i.
\]
Thus, the finite-dimensional Kantorovich dual characterization remains
unchanged, while the primal assignment is represented by a stochastic
kernel rather than, in general, a deterministic transport map.



\begin{IEEEbiography}[{\includegraphics[width=1in,height=1.25in,clip,keepaspectratio]{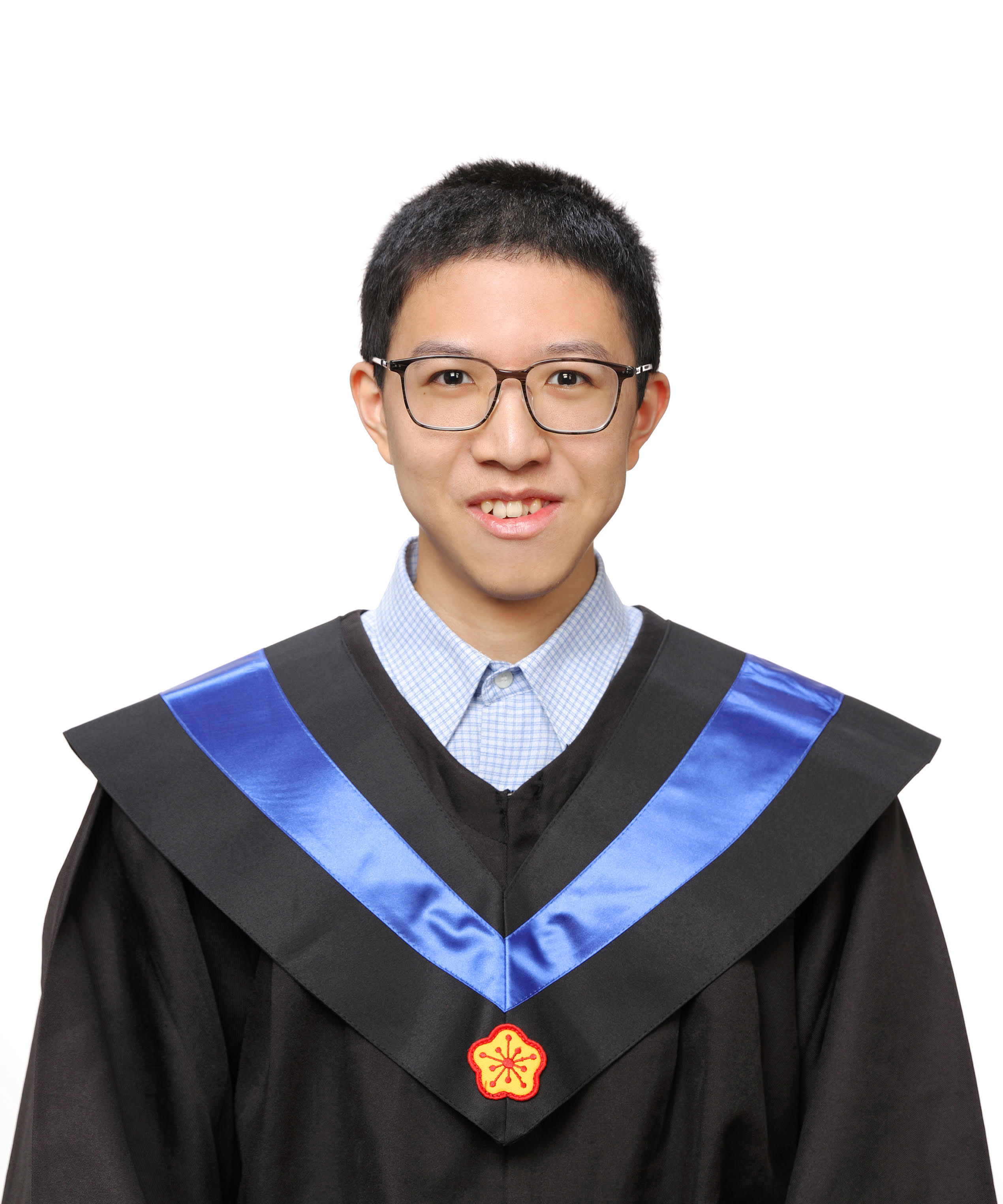}}]{Pan-Yang Su} (Student Member, IEEE) received the B.S. degree in electrical engineering from National
Taiwan University, Taipei, Taiwan, in 2022. He
is currently pursuing the Ph.D. degree with the
Department of Electrical Engineering and Computer
Sciences, University of California at Berkeley,
Berkeley, CA, USA.

His research focuses on applying game theory and mechanism design techniques to study multi-agent systems, with applications in energy and transportation.
\end{IEEEbiography}

\begin{IEEEbiography}[{\includegraphics[width=1in,height=1.25in,clip,keepaspectratio]{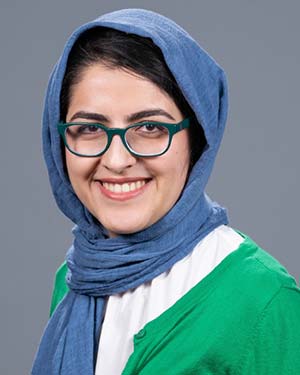}}]{Negar Mehr} (Member, IEEE) received the B.S. degree in mechanical engineering from Sharif University of Technology, Tehran, Iran, in 2013, and the Ph.D. degree in mechanical engineering from the University of California, Berkeley, Berkeley, CA, USA, in 2019. 

From 2019 to 2020, she was a Postdoctoral Scholar in Aeronautics and Astronautics at Stanford University, Stanford, CA, USA. She then served as an Assistant Professor of Aerospace Engineering at the University of Illinois Urbana-Champaign, Urbana, IL, USA, before joining the University of California, Berkeley, in 2024. She is currently an Assistant Professor of Mechanical Engineering and directs the Berkeley Intelligent Control (ICON) Laboratory. Her research interests include control theory, game theory, and machine learning, with an emphasis on safe and intelligent interactions among robots, humans, and other autonomous agents.

Dr. Mehr is a recipient of the AFOSR Young Investigator Award, the ONR Young Investigator Award, the NSF CAREER Award, the ASME Rising Star Award, and the IEEE ITSS Best Ph.D. Dissertation Award.
\end{IEEEbiography}

\begin{IEEEbiography}[{\includegraphics[width=1in,height=1.25in,clip,keepaspectratio]{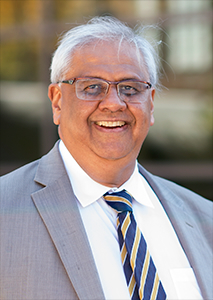}}]{Shankar Sastry} (Life Fellow, IEEE)  is a Professor of EECS, Bioengineering and Mechanical Engineering and the Siebel Professor of Computer Science and
Engineering at the University of California Berkeley.
He was the Dean of Engineering from 2007-18. He
is currently the faculty Director of the Technology
Competitiveness and Industrial Policy, a center
which conducts studies on policy need to enhance
US competitiveness and the FHL Vive Center for Enhanced Reality for exploring the boundaries
between augmented reality and real scenes with applications to autonomy, performance arts, and
education. He is the honorary chancellor of a new technical university in Chandigarh, India called
Plaksha University. He is deeply committed to harnessing the power of technology to lift people
out of poverty through his directorship of the Blum Center for Developing Economies 2007-22.

He has coauthored about 750 technical papers. He has co-authored or co-edited 10 books, including
Adaptive Control: Stability, Convergence and Robustness (with M. Bodson, Prentice Hall, 1989)
and A Mathematical Introduction to Robotic Manipulation (with R. Murray and Z. Li, CRC Press,
1994), Nonlinear Systems: Analysis, Stability and Control (Springer-Verlag, 1999), and An
Invitation to 3D Vision: From Images to Models (Springer Verlag, 2003) (with Y. Ma. S. Soatto,
and J. Kosecka). and Generalized Principal Component Analysis (Springer Verlag, 2016 with R.
Vidal and Y. Ma). 
Dr. Sastry was elected to the National Academy of Engineering in 2001 the
American Academy of Arts and Sciences in 2004. He received the President of India Gold Medal
in 1977, the NSF Presidential Young Investigator Award in 1985. He got the Eckman Award of the
American Automatic Control Council in 1990, the Ragazzini Award for Distinguished
Accomplishments in teaching in 2005 and the Rufus Oldenburger Career Award of ASME in 2021.
He has received honorary doctorates from KTH, the Royal Institute of Technology, Stockholm, a
Ph.D. honoris causa from the University of Waterloo, Canada, and a Laurea Dottorato honoris
causa from Politecnico di Torino, and the Berkeley Citation. He has supervised more than 75
doctoral students to completion and over 50 Masters of Science Students and is very proud of the
leadership positions that they hold in academia and industry.
\end{IEEEbiography}

\end{document}